\documentclass[onecolumn]{article}

\usepackage[numbers,sort,compress]{natbib}

\usepackage[letterpaper, left=1in, right=1in, top=1in, bottom=1in]{geometry}
\usepackage{graphicx}
\usepackage{multirow}
\usepackage{amsmath,amssymb,amsfonts}
\usepackage{amsthm}
\usepackage{mathrsfs}
\usepackage[title]{appendix}
\usepackage{xcolor}
\usepackage{textcomp}
\usepackage{manyfoot}
\usepackage{booktabs}
\usepackage{array}

\theoremstyle{thmstyleone}
\newtheorem{theorem}{Theorem}

\theoremstyle{thmstyletwo}
\newtheorem{example}{Example}

\theoremstyle{thmstylethree}

\newlength{\thickarrayrulewidth}
\newcommand{\dd}{\mathrm{d}}
\newcommand{\Ham}{\mathcal{H}}
\newcommand{\Lag}{\mathcal{L}}
\newcommand{\Sym}{\operatorname{Sym}}
\newcommand{\id}{\operatorname{id}}
\newcommand{\BK}{\mathrm{BK}}
\newcommand{\eps}{\varepsilon}
\newcommand{\FS}{{\mathrm{FS}}}
\begin{document}

\title{Gauge-covariant Hamilton equations for teleparallel equivalents of general relativity}

% Submit to Class and Quant Gravity.

\author{David Chester\thanks{Quantum Gravity Research, Topanga, CA, USA; email: \texttt{davidc@quantumgravityresearch.org}}\,\, and 
Vipul Pandey\thanks{Department of Physics, School of Sciences, IILM University, Greater Noida - 201306, India; email: \texttt{vipul.pandey@iilm.edu}}}

\maketitle
\setcounter{page}{1}

\begin{abstract}
We construct a gauge-covariant Hamiltonian formulation for the metric and symmetric teleparallel equivalents of general relativity. When treating torsion and nonmetricity as generalized velocities, the conjugate momenta are proportional to their respective superpotentials. The quadratic actions give regular Legendre maps on these field-strength fibers. For both theories, the second Hamilton equation reproduces the Einstein field equations. In the reduced ADM sector, the hypersurface generators obey the standard deformation algebra. The regular field-strength map changes the Legendre description but does not remove the gravitational constraints or alter the classical physics.
\end{abstract}

%\maketitle

\pagebreak

\tableofcontents

\pagebreak

%======================================================================
\section{Introduction}
\label{sec:introduction}

General relativity admits equivalent descriptions based on curvature,
torsion, or nonmetricity, commonly called the geometrical trinity of gravity
\cite{BeltranJimenez:2019tjy,Capozziello:2021pcg}. The torsion formulation was
developed into the metric teleparallel equivalent of general relativity
(MTEGR) \cite{Sauer:2004hj,1976Cho,Hayashi1979,Maluf:1994ji,Maluf:2013gaa}.
Its gauge structure has also been studied within Poincar\'e gauge gravity
\cite{Blagojevic:2000pi,Blagojevic:2003cg,Barker2023}. The symmetric
teleparallel equivalent (STEGR) uses nonmetricity with vanishing curvature and
torsion \cite{Nester:1998mp,BeltranJimenez:2017tkd}. Both theories fit within
metric-affine gauge gravity
\cite{Hehl:1976my,HEHL19951,obukhov2003metric,Capozziello:2022zzh}.

The Einstein--Hilbert action is linear in curvature, so its covariant
field-strength Legendre map is singular. The quadratic MTEGR and STEGR actions
describe the same gravity with a different Legendre structure. Torsion and
nonmetricity can serve as generalized velocities, with the teleparallel
superpotentials as their conjugate momenta. De Donder--Weyl mechanics provides
a covariant setting without choosing a preferred foliation
\cite{Gunther1987,Kanatchikov:1997wp,Kanatchikov:2000jz}.

Canonical studies establish the MTEGR constraint structure and give $3+1$
formulations of STEGR
\cite{Maluf:1994ji,Maluf:2000zt,Guzman:2023,CapozzielloSauro}. The STEGR
construction of Ref.~\cite{Guzman:2023} uses the coincident gauge. A De
Donder--Weyl treatment of TEGR starts from tetrad Palatini variables and uses
generalized Dirac brackets \cite{Kanatchikov:2023tegr}. We instead construct a
common gauge-covariant formulation directly from the complete torsion and
nonmetricity field strengths. Their momenta are the teleparallel
superpotentials, and their second Hamilton equations display the Einstein
equation. For STEGR, Eqs.~\eqref{eq:C-action} and
\eqref{eq:stegr-phase-einstein} give the full $40$-component inverse and the
Einstein tensor in phase-space variables. Neither requires a $3+1$ split or
the coincident gauge.

The STEGR equivalence theorem applies to the metric sector at fixed flat
torsionless inertial connection, with boundary conditions that remove the
metric surface term. Legendre regularity concerns the field-strength fibers,
not the full gauge phase space. Diffeomorphism, connection, multiplier,
boundary, and global constraints remain. We do not give their unreduced Dirac
classification.

The metric-sector Hamiltonian structure can be pulled back to a hypersurface
using covariant phase-space methods
\cite{Crnkovic1988CPS,Lee:1990nz,Iyer:1994ys}. In the reduced ADM sector, this
gives the presymplectic form, the deformation generator, and its corner flux.
Its closure is governed by the field-dependent hypersurface-deformation
bracket \cite{Bergmann:1972ww,Teitelboim1973,Hojman:1976vp,Blohmann:2010jd}.
The independent connection and multiplier sectors are not included in that
reduced generator.

\subsection{Summary of results}

The principal result is a gauge-covariant field-strength Hamiltonian
formulation of MTEGR and STEGR. Torsion and nonmetricity are the generalized
velocities, and the teleparallel superpotentials are their momenta. For STEGR,
the inverse \eqref{eq:C-action} gives the quadratic Hamiltonian
\eqref{eq:stegr-ham} on all $40$ components. The second Hamilton equation
yields the Einstein-form identity \eqref{eq:stegr-phase-einstein}.
Invertibility is proved in Sec.~\ref{ssec:stegr} and App.~\ref{app:inverse}.
Theorem~\ref{thm:stegr-equivalence} proves a two-way equivalence between the
bulk metric Euler--Lagrange equation and the first-order Hamilton equations at
fixed flat torsionless inertial connection.

The second Hamilton equation gives the Einstein equation for every fixed flat
torsionless inertial connection in the declared STEGR metric sector. The MTEGR
inverse \eqref{eq:mtegr-inverse} and Einstein identity
\eqref{eq:mtegr-einstein} give the corresponding result in the torsion sector.
These results concern the covariant bulk equations. In the compact,
boundaryless, boundary-subtracted, trivial-holonomy coincident-gauge sector,
the separate canonical reduction gives the first-class ADM constraints with
two local degrees of freedom (App.~\ref{app:dirac}). This result is compared
with the existing MTEGR and STEGR Hamiltonian analyses
\cite{Maluf:2000zt,Guzman:2023,CapozzielloSauro}. It does not determine the
constraint matrix of the unreduced connection and multiplier theory.

Pulling the metric-sector polysymplectic form to a hypersurface gives the
reduced deformation generator and its conditional Bergmann--Komar algebra
(Sec.~\ref{sec:hypersurface} and App.~\ref{app:generator})
\cite{Bergmann:1972ww,Teitelboim1973,Blohmann:2010jd}. The
field-strength Hamiltonian is one term in the raw normal density. On shell,
the reduced generator becomes a corner integral. Bulk evolution on a closed
universe remains gauge, as in canonical general relativity.
In compact diagonal Bianchi I, swept four-volume exactly reparametrizes the
flow of $H_{\rm BI}=N\mathcal C_{\rm BI}$
(Example~\ref{ex:bianchi-volume-time}).

\subsection{Conventions}
\label{ssec:conventions}

We use signature $(-,+,+,+)$, $\kappa^2=32\pi G/c^4$, and weighted symmetrization $A_{(\mu\nu)}=\tfrac12(A_{\mu\nu}+A_{\nu\mu})$ throughout (where a different convention briefly appears in a quoted closed form, we say so explicitly). The last lower connection index is the derivative slot:
\begin{equation}
 \nabla_\mu V^\rho=\partial_\mu V^\rho+\Gamma^\rho{}_{\sigma\mu}V^\sigma,
 \qquad
 \nabla_\mu w_\rho=\partial_\mu w_\rho-\Gamma^\sigma{}_{\rho\mu}w_\sigma.
 \label{eq:connection-slot-convention}
\end{equation}
With this convention, nonmetricity, the paper's torsion sign, and curvature are
\begin{align}
Q_{\rho\mu\nu}
&=\partial_\rho g_{\mu\nu}
-\Gamma^\lambda{}_{\mu\rho}g_{\lambda\nu}
-\Gamma^\lambda{}_{\nu\rho}g_{\mu\lambda}, \nonumber\\
T^\rho{}_{\mu\nu}
&=2\Gamma^\rho{}_{[\mu\nu]}, \qquad
R^\rho{}_{\sigma\mu\nu}
=2\partial_{[\mu}\Gamma^\rho{}_{|\sigma|\nu]}
+2\Gamma^\rho{}_{\lambda[\mu}\Gamma^\lambda{}_{|\sigma|\nu]}.
\end{align}
so $T^\rho{}_{\mu\nu}$ is the negative of the torsion components used by authors who combine the same derivative-slot convention with $2\Gamma^\rho{}_{[\nu\mu]}$. All formulas below, including the MTEGR superpotential, use the displayed sign consistently.
The nonmetricity traces are $Q_\rho=Q_{\rho\lambda}{}^\lambda$ and
$\widetilde Q_\rho=Q_{\lambda\rho}{}^\lambda$, with analogous momentum traces
$\Pi_\rho=\Pi_{\rho\lambda}{}^\lambda$ and
$\widetilde\Pi_\rho=\Pi_{\lambda\rho}{}^\lambda$. The torsion trace is
$T_\mu=\delta_\rho^\nu T^\rho{}_{\mu\nu}$.
The affine connection decomposes as $\Gamma^\rho{}_{\mu\nu}=\bar\Gamma^\rho{}_{\mu\nu}+N^\rho{}_{\mu\nu}$ with $\bar\Gamma$ the Levi-Civita connection and the distortion $N^\rho{}_{\mu\nu}=-K^\rho{}_{\mu\nu}-L^\rho{}_{\mu\nu}$ split into contortion $K^\rho{}_{\mu\nu}$ from torsion and disformation $L^\rho{}_{\mu\nu}$ from nonmetricity,
\begin{align}
\bar{\Gamma}_{\rho\mu\nu}&= \partial_{(\mu} g_{|\rho|\nu)} -\tfrac12 \partial_\rho g_{\mu\nu}, \nonumber \\
K^\rho{}_{\mu\nu}&=T_{(\mu}{}^\rho{}_{\nu)} - \tfrac12 T^\rho{}_{\mu\nu},  \\
L^\rho{}_{\mu\nu}&=Q_{(\mu}{}^\rho{}_{\nu)} - \tfrac12 Q^\rho{}_{\mu\nu} \nonumber .
\end{align}
Bars denote pseudo-Riemannian objects with the Levi-Civita connection ($\bar\nabla$, $\bar R$). The local gauge fields are $(g_{ab},e^a{}_\mu,\omega_\mu{}^a{}_b)$ with corresponding local field strengths $(\Xi_{\mu ab},\Theta_{\mu\nu}{}^a,\Omega_{\mu\nu}{}^a{}_b)$. The tetrad postulate $\nabla_\mu e^a{}_\nu=0$ maps between the local and global connections. A summary of metric-affine and hypersurface symbols is collected in App.~\ref{app:conventions}. We reserve $\Ham$ for Hamiltonian densities of weight zero and write $\sqrt{-g}\,\Ham$ for the weight-one density.

\subsection{The teleparallel actions}
\label{ssec:actions}

The STEGR and MTEGR actions follow by applying different constraints to general teleparallel geometry and expressing the Einstein--Hilbert action in terms of a single field strength \cite{BeltranJimenez:2019tjy,Capozziello:2022zzh}. The STEGR action is obtained by evaluating the pseudo-Riemannian Ricci scalar in a symmetric-teleparallel spacetime with nonmetricity, while torsion and curvature vanish \cite{Nester:1998mp,BeltranJimenez:2017tkd},
\begin{align}
S_{\textrm{STEGR}}
&= \left. -\frac{c^4}{16\pi G}\int d^4x\,\sqrt{-g}\,\bar R \right|_{R=T=0}
=-\frac{2}{\kappa^2}\int d^4x\,\sqrt{-g}\Bigl[
\bar\nabla_\alpha\bigl(\widetilde Q^\alpha-Q^\alpha\bigr) \nonumber\\
&\hspace{2.1cm}
-\frac14 Q_{\rho\mu\nu}Q^{\rho\mu\nu}
+\frac12 Q_{\rho\mu\nu}Q^{\mu\nu\rho}
+\frac14 Q_\mu Q^\mu-\frac12 Q_\mu\widetilde Q^\mu\Bigr].
\label{eq:stegr-action}
\end{align}
Dropping the total-derivative term does not change the classical equations of motion. The resulting bulk action is proportional to the nonmetricity scalar,
\begin{align}
S_{\textrm{STEGR}} &\approx -\frac{2}{\kappa^2}\int d^4x\,\sqrt{-g}\,\mathbb Q
=-\frac{1}{\kappa^2}\int d^4x\,\sqrt{-g}\,Q_{\rho\mu\nu}\mathfrak P^{\rho\mu\nu}, \\
\mathbb Q &=-\frac14 Q_{\rho\mu\nu}Q^{\rho\mu\nu}+\frac12 Q_{\rho\mu\nu}Q^{\mu\nu\rho}+\frac14 Q_\mu Q^\mu-\frac12 Q_\mu\widetilde Q^\mu
=\frac12 Q_{\rho\mu\nu}\mathfrak P^{\rho\mu\nu}.
\end{align}
Here the nondensitized nonmetricity superpotential is denoted by
\begin{equation}
\mathfrak P^{\rho\mu\nu}=-\tfrac12 Q^{\rho\mu\nu}+\tfrac12 Q^{\mu\nu\rho}+\tfrac12 Q^{\nu\mu\rho}+\tfrac12\bigl(Q^\rho-\widetilde Q^\rho\bigr)g^{\mu\nu}-\tfrac14\bigl(g^{\rho\mu}Q^\nu+g^{\rho\nu}Q^\mu\bigr).
\label{eq:stegr-superpotential}
\end{equation}
The boundary term is essential for equivalence at the level of conserved charges \cite{Nester:1998mp,BeltranJimenez:2019tjy}, and it will reappear in the hypersurface analysis of Sec.~\ref{ssec:charges}. 

The MTEGR action is similarly proportional to the Ricci scalar in a metric-teleparallel spacetime with torsion \cite{Hayashi1979,Maluf:2013gaa,Capozziello:2022zzh},
\begin{align}
S_{\textrm{MTEGR}} &= \left.-\frac{c^4}{16\pi G}\int d^4x\,|e|\,\bar R \right|_{R = Q=0} \label{eq:mtegr-action} \\
 &= -\frac{c^4}{16\pi G}\int d^4x\,|e|\left(2\bar\nabla_\mu T^\mu-\frac14 T_{\rho\mu\nu}T^{\rho\mu\nu}+\frac12 T_{\rho\mu\nu}T^{\mu\nu\rho}+T_\mu T^\mu\right), \nonumber
\end{align}
conveniently factorized through the torsion superpotential $S_\rho{}^{\mu\nu}=-K^{\mu\nu}{}_\rho+T^{\mu}\delta^{\nu}_\rho-T^{\nu}\delta^{\mu}_\rho$ as 
\begin{equation}
S_{\textrm{MTEGR}}\approx-\frac{1}{\kappa^2}\int d^4x\,|e|\,T^\rho{}_{\mu\nu}S_\rho{}^{\mu\nu}.
\end{equation} 
Both actions are quadratic in their respective field strengths. Their
field-strength Legendre maps are therefore algebraic, and
Secs.~\ref{ssec:stegr} and~\ref{ssec:mtegr} establish their invertibility on
the declared fibers.

Section~\ref{sec:classical} constructs the field-strength Legendre maps and
Hamilton equations for metric-affine gravity.
Section~\ref{sec:teleparallel} specializes them to
STEGR and MTEGR, gives the explicit inverses, and derives the Einstein-form
equations. Section~\ref{sec:hypersurface} derives the reduced hypersurface generator and
deformation algebra, relates it to covariant Hamiltonian evolution, and
identifies the lapse zero mode measured by swept four-volume.
Section~\ref{sec:conclusions} gives the final scope.
Appendices A--D collect the notation, polysupermetric inverse, generator
identity, and reduced Dirac analysis.

%======================================================================
\section{Gauge-covariant Hamiltonians for metric-affine gravity}
\label{sec:classical}

\subsection{The field-strength Hamiltonian formulation}
\label{ssec:bundles}\label{ssec:fsform}

De Donder--Weyl mechanics promotes the partial derivatives $\partial_\mu\phi^A$ to independent polyvelocity coordinates on the first jet bundle $J^1Y$ of a configuration bundle $\pi:Y\to X$ \cite{Gunther1987,Kanatchikov:1997wp,Kanatchikov:1997pj,Kanatchikov:2000jz,Carinena:1991,Gotay:1997,Forger:2004}. The present formulation instead takes the fiberwise Legendre transform with respect to a specified field strength. This requires identifying the field-strength bundle and its relation to the first jet.

A field, or gauge potential, $\phi^A$ is a section of the configuration bundle
$\pi:Y\to X$ over spacetime $X$. The index $A$ labels its components in a
local trivialization. Examples are the gauge potential $A_\mu^a$, the coframe
$e^a{}_\mu$, and the metric $g_{\mu\nu}$, with field strengths
$F_{\mu\nu}^a$, $T^a{}_{\mu\nu}$, and $Q_{\rho\mu\nu}$. The field strength
$F_\mu{}^A$ is the first-order differential quantity through which the
Lagrangian depends on derivatives of $\phi^A$.

In general, the field strength is a natural first-order bundle map, not the derivative of the potential itself. In local jet coordinates, every field strength used here is affine in the first derivative,
\begin{equation}
F_\mu{}^A=m_\mu{}^{A}{}_{B}{}^{\nu}(\phi,x)\,
\partial_\nu\phi^B+b_\mu{}^A(\phi,x),
\label{eq:fs-affine-def}
\end{equation}
where $m$ is the principal symbol and $b$ is the zero-jet affine part. The latter contains connection actions and self-interactions, such as $A\wedge A$ in a nonabelian curvature. Neither term is generally tensorial by itself. Their sum is the specified natural bundle map. Only when $m$ is the identity and $b$ is the connection action does Eq.~\eqref{eq:fs-affine-def} reduce to a covariant polyvelocity $F_\mu{}^A=\nabla_\mu\phi^A$. This includes scalar fields and STEGR nonmetricity, but not torsion or curvature as field strengths of their potentials.

Given a configuration bundle $Y\to X$ and a specified smooth bundle map over $X$,
\begin{equation}
\mathcal F:\;J^1Y\longrightarrow E_F,
\end{equation}
into a specified vector bundle $E_F\to X$ of field-strength type, the field-strength variables are coordinates on the fiber product $Z_F:=Y\times_X E_F$ via the graph map
\begin{equation}
\widehat{\mathcal F}:\;J^1Y\to Z_F,\qquad
j^1_x\phi\mapsto\bigl(\phi(x),\,\mathcal F(j^1_x\phi)\bigr),
\label{eq:graph-map}
\end{equation}
subject to the compatibility relation $F=\mathcal F(j^1\phi)$ on sections. A bare coordinate expression $F_\mu{}^A(\phi,\partial\phi;x)$ does not by itself define a bundle. On a rank-$N$ vector bundle with transition functions $M^A{}_B(x)$, the derivative transforms inhomogeneously,
\begin{equation}
\partial_{\rho'}\phi^A_{(j)}
= \frac{\partial x^\rho}{\partial x^{\rho'}}\Bigl(M^A{}_B\,\partial_\rho\phi^B_{(i)}+(\partial_\rho M^A{}_B)\,\phi^B_{(i)}\Bigr).
\label{eq:overlap}
\end{equation}
Local formulas glue to a global $\mathcal F$ only when built from natural operations, such as an exterior derivative or covariant derivative with connection/gauge field, on specified bundles. All cases used in this paper are of this type, and for each one the principal symbol of $\mathcal F$, its linearization in the first-derivative jet coordinates, explains which first-jet data the field strength forgets.

For gauge theories and gravity the resulting field strengths live on the following bundles.

\begin{example}[scalar field]
$Y=X\times\mathbb R$, $E_F=T^*X$, $\mathcal F(j^1\varphi)=\dd\varphi$. The field strength remembers the entire first jet. This is the canonical DDW case, with $m=\id$ and $b=0$ in Eq.~\eqref{eq:fs-affine-def}. A complex scalar of charge $q$ instead has $m=\id$ and $b_\nu=-iqA_\nu\varphi$, so its field strength is the minimal-coupling combination $D_\nu\varphi$.
\end{example}

\begin{example}[Yang--Mills]
\label{ex:ym}
$Y=C(P)$, the affine bundle of connections on a principal $G$-bundle, $E_F=\Lambda^2T^*X\otimes\mathrm{ad}(P)$, and
$F^a{}_{\mu\nu}=A^a{}_{\nu,\mu}-A^a{}_{\mu,\nu}+g_{\rm YM}f^a{}_{bc}A^b{}_\mu A^c{}_\nu$. The symbol is the antisymmetrization map. Its kernel consists of the symmetric part of the potential's first jet, i.e.\ exactly the linearized gauge directions $\delta A=\dd\lambda$ at the symbol level. This is the precise version of the statement that the curvature ``defines a field-strength bundle.''
\end{example}

\begin{example}[MTEGR \& torsion]
$Y=\mathrm{Cof}(V)\times_X C_{\rm SO}(P_{\rm SO})$ (coframes $\times$ Lorentz connections), $E_T=\Lambda^2T^*X\otimes V$, $\mathcal F_T(j^1(e,\omega))=d_\omega e$. The symbol acts only through the coframe jet, since torsion sees no derivatives of the connection. The teleparallel restriction $\Omega(\omega)=0$ is therefore additional data, imposed by restriction or multipliers, independent of the torsion map.
\end{example}

\begin{example}[STEGR \& nonmetricity]
\label{ex:stegrbundle}
$Y=\mathrm{Lor}(T^*X)\times_X C(TX)$ (Lorentzian metrics $\times$ linear connections), $E_Q=T^*X\otimes\Sym^2T^*X$, and $Q_{\rho\mu\nu}=\nabla_\rho g_{\mu\nu}$. The first-derivative linearization is
\begin{equation}
\frac{\partial Q_{\rho\mu\nu}}{\partial g_{\alpha\beta,\sigma}}=\delta^\sigma_\rho\,\delta^\alpha_{(\mu}\delta^\beta_{\nu)},
\qquad
\frac{\partial Q_{\rho\mu\nu}}{\partial\Gamma^\lambda{}_{\alpha\beta,\sigma}}=0.
\end{equation}
The symbol is injective on the metric first jet. Nonmetricity has the same tensor type as $\partial g$, and no metric first-jet direction lies in the kernel. The kernel consists only of connection first-derivative increments. This separates the field-strength Legendre map on the nonmetricity fiber from the independent restrictions $T(\Gamma)=0$ and $R(\Gamma)=0$.
\end{example}

Global metric-affine gravity with $(g,\Gamma)$ and all three field strengths
leads to the direct sum $E_{\rm MAG}=E_Q\oplus E_T\oplus E_R$. Nonmetricity
contains all metric first derivatives, curvature contains only the
antisymmetric connection derivatives, and torsion contains no connection
derivatives. Thus $(Q,T,R)$ do not form a bundle isomorphic to $J^1Y$ without
extra data. Potentials and field strengths instead give coordinates on the
fiber product $Y\times_X E_{\rm MAG}$ through the graph of
$\widehat{\mathcal F}$. The corresponding local potentials are
$(g_{ab},e_\mu{}^a,\omega_\mu{}^a{}_b)$. Table~\ref{tab:DDW_MAG} relates their
momenta to DDW polymomenta and to the excitations of electrodynamics
\cite{Hehl2003,Puetzfeld:2014qba,Obukhov:2018bmf}.

\begin{table}[t]
\centering
\small
\begin{tabular}{@{}p{0.23\linewidth}p{0.23\linewidth}p{0.42\linewidth}@{}}
\toprule
De Donder--Weyl & Electrodynamics & Metric-affine gravity \\
\midrule
Field $\phi^A$
& Potential $A_\mu$
& Potentials $(g_{ab},e^a{}_\mu,\omega_\mu{}^a{}_b)$ \\
Polyvelocity $\partial_\mu\phi^A$
& Field strength $F_{\mu\nu}$
& Field strengths $(Q_{\rho\mu\nu},T^\rho{}_{\mu\nu},
R^\rho{}_{\sigma\mu\nu})$ \\
Polymomentum
$\partial(\sqrt{-g}\Lag)/\partial(\partial_\mu\phi^A)$
& Excitation
$H^{\mu\nu}=\partial(\sqrt{-g}\Lag)/\partial F_{\mu\nu}$
& Momenta
$(\Pi,P,\Sigma)=\partial(\sqrt{-g}\Lag)/\partial(Q,T,R)$ \\
\bottomrule
\end{tabular}
\caption{Dictionary between DDW mechanics, macroscopic electrodynamics, and metric-affine gauge gravity. The excitation contains the polarization and magnetization. The gravitational excitations are proportional to the teleparallel superpotentials.}
\label{tab:DDW_MAG}
\end{table}

With the field-strength bundle established, the Hamiltonian formulation is its fiberwise Legendre transform, taken in the field strength instead of in a time derivative. For a Lagrangian density $L=\sqrt{-g}\,\Lag(\phi^A,F_\mu{}^A;x)$ the momentum field strength is
\begin{equation}
\Pi^\mu{}_A=\frac{\partial\sqrt{-g}\,\Lag}{\partial F_\mu{}^A},
\end{equation}
and the field-strength Hamiltonian density is the fiberwise Legendre transform
\begin{equation}
\sqrt{-g}\,\Ham=\Pi^\mu{}_A F_\mu{}^A-\sqrt{-g}\,\Lag .
\label{eq:legendre}
\end{equation}

Differentiating Eq.~\eqref{eq:legendre} with the chain rule gives two Legendre identities that carry the whole formulation. At fixed fields,
\begin{equation}
\frac{\partial\sqrt{-g}\,\Ham}{\partial\Pi^\mu{}_A}
=F_\mu{}^A
+\Pi^\nu{}_B\frac{\partial F_\nu{}^B}{\partial\Pi^\mu{}_A}
-\frac{\partial\sqrt{-g}\,\Lag}{\partial F_\nu{}^B}\frac{\partial F_\nu{}^B}{\partial\Pi^\mu{}_A}
=F_\mu{}^A,
\label{eq:legendre-id1}
\end{equation}
and at fixed momenta,
\begin{equation}
\frac{\partial\sqrt{-g}\,\Ham}{\partial\phi^A}\bigg|_{\Pi}
=\Pi^\nu{}_B\frac{\partial F_\nu{}^B}{\partial\phi^A}
-\frac{\partial\sqrt{-g}\,\Lag}{\partial\phi^A}\bigg|_{F}
-\frac{\partial\sqrt{-g}\,\Lag}{\partial F_\nu{}^B}\frac{\partial F_\nu{}^B}{\partial\phi^A}
=-\,\frac{\partial\sqrt{-g}\,\Lag}{\partial\phi^A}\bigg|_{F},
\label{eq:legendre-id2}
\end{equation}
where in each line the middle terms cancel against the definition of the momentum. The first identity inverts the momentum map, while the second converts the fixed-momentum gradient of the Hamiltonian density into the fixed-field-strength gradient of the Lagrangian density. The Euler--Lagrange equations are therefore equivalent to the pair of field-strength Hamilton equations whenever the Legendre map is invertible on the field-strength fiber, with no reference to a foliation. The same equivalence follows from varying the canonical action, as displayed next.

For the affine first-jet map in Eq.~\eqref{eq:fs-affine-def}, define
\begin{align}
 \mathscr E_{\mathcal F,A}(\Pi)
 &:={\partial}_\nu\!\left(
 m_\mu{}^{B}{}_{A}{}^{\nu}\Pi^\mu{}_B\right)
 -\left[
 \frac{\partial m_\mu{}^{B}{}_{C}{}^{\nu}}{\partial\phi^A}
 \partial_\nu\phi^C
 +\frac{\partial b_\mu{}^B}{\partial\phi^A}
 \right]\Pi^\mu{}_B .
 \label{eq:affine-euler-operator}
\end{align}
The operator $\mathscr E_{\mathcal F}$ is the negative formal adjoint of the linearized field-strength map acting on the momentum. For differential-form field strengths, its derivative part is the covariant exterior derivative of the excitation form. Its remaining terms account for the algebraic dependence of the field strength on the potentials. Variation of the canonical action $S=\int d^4x\,(\Pi^\mu{}_AF_\mu{}^A-\sqrt{-g}\,\Ham)$ gives the two field-strength Hamilton equations in their general form,
\begin{equation}
F_\mu{}^A=\frac{\partial\sqrt{-g}\,\Ham}{\partial\Pi^\mu{}_A},
\qquad
\mathscr E_{\mathcal F,A}(\Pi)
=-\frac{\partial\sqrt{-g}\,\Ham}{\partial\phi^A},
\label{eq:fs-hamilton}
\end{equation}
which reproduce the Euler--Lagrange equations whenever the Legendre map is invertible on the field-strength fiber. The first equation returns the field strength itself, not $\nabla_\mu\phi^A$ in general. The second contains the complete affine operator $\mathscr E_{\mathcal F}$, not a covariant divergence in general. Both reduce to the familiar covariant-polyvelocity equations $F_\mu{}^A=\nabla_\mu\phi^A$ and $\nabla_\mu\Pi^\mu{}_A=-\partial(\sqrt{-g}\Ham)/\partial\phi^A$ only when the symbol is the identity and $b$ is the connection action on $\phi$. Scalar fields and STEGR nonmetricity have this form. Torsion and curvature instead have antisymmetrizing symbols, so their equations retain the corresponding projection encoded by $\mathscr E_{\mathcal F}$.
These equations may also be expressed as Poisson--Gerstenhaber derivations of
Hamiltonian forms \cite{Kanatchikov:1997wp,Kanatchikov:2000jz}, but no
graded-bracket construction is needed below.

The Maxwell example below shows how the field-strength Hamilton equations
reproduce the field equations without a time split. The local and global
metric-affine equations are summarized in App.~\ref{app:conventions}. For
quadratic theories,
\begin{equation}
\Lag(\phi,F;x)=\tfrac12 M^{\mu\nu}{}_{AB}(\phi;x)\,F_\mu{}^AF_\nu{}^B-V(\phi;x),
\end{equation}
the object $M^{\mu\nu}{}_{AB}$ is a metric on the space of field strengths. We call it the ``polysupermetric'', since it generalizes Wheeler's supermetric on superspace \cite{Wheeler1968} to the polymomentum setting. When the polysupermetric is invertible on the actual field-strength fiber, the Hamiltonian density is again quadratic,
\begin{equation}
\Ham(\phi,\Pi;x)=\tfrac12 M_{\mu\nu}{}^{AB}(\phi;x)\,\Pi^\mu{}_A\Pi^\nu{}_B+V(\phi;x),
\end{equation}
with $M_{\mu\nu}{}^{AB}$ the fiberwise inverse, and on constrained tensor fibers the correct identity defining the inverse is contraction to the projector onto the fiber. 

\begin{example}[electrodynamics]
\label{ex:em}
Take Minkowski spacetime with the potential $A_\mu$ as position field, so the velocity field strength is the electromagnetic field strength $F_{\mu\nu}=\partial_\mu A_\nu-\partial_\nu A_\mu$ and the field-strength action is
\begin{equation}
S(A_\mu,F_{\mu\nu};x)=\int d^4x\left(-\frac{1}{4\mu_0}F_{\mu\nu}F^{\mu\nu}-j^\mu A_\mu\right).
\end{equation}
The momentum field strength,
\begin{equation}
\Pi^{\mu\nu}=\frac{\partial\Lag}{\partial F_{\mu\nu}}=-\frac{1}{2\mu_0}F^{\mu\nu},
\end{equation}
is proportional to the excitation tensor of macroscopic electrodynamics, which packages the electric displacement $\vec D$ and the magnetic field $\vec H$ \cite{Hehl2003}. The field-strength Hamiltonian density is
\begin{equation}
\Ham=-\mu_0\,\Pi_{\mu\nu}\Pi^{\mu\nu}+j^\mu A_\mu,
\end{equation}
and the field-strength Hamilton equations read
\begin{equation}
F_{\mu\nu}=\frac{\partial\Ham}{\partial\Pi^{\mu\nu}}=-2\mu_0\,\Pi_{\mu\nu},
\qquad
\partial_\nu\bigl(\Pi^{\nu\mu}-\Pi^{\mu\nu}\bigr)=-\frac{\partial\Ham}{\partial A_\mu}=-j^\mu .
\end{equation}
Substituting the first into the second returns Maxwell's equations, $\partial_\nu F^{\mu\nu}=-\mu_0 j^\mu$. The field-strength formulation bypasses the Legendre constraints and Dirac--Bergmann stabilization procedure needed when Maxwell theory is embedded in the conventional potential-based canonical phase space. It does not remove gauge symmetry. Gauss's law follows directly as the normal component of the covariant Hamilton equation.
\end{example}

This pattern of potential, field strength, excitation, then quadratic Hamiltonian density is the one followed by the gravitational sectors, as summarized in Table~\ref{tab:DDW_MAG}.

\subsection{Local metric-affine gravity}\label{ssec:localglobal}

The field-strength formulation of Sec.~\ref{ssec:fsform} was stated for a generic bundle map $\mathcal F:J^1Y\to E_F$. This subsection instantiates it for metric-affine gravity in both the local triplet $(g_{ab},e^a{}_\mu,\omega_\mu{}^a{}_b)$ and the global pair $(g_{\mu\nu},\Gamma^\rho{}_{\mu\nu})$. It also includes the geometric identities that reduce the metric-affine structure to the teleparallel equivalents treated in Secs.~\ref{ssec:mtegr} and~\ref{ssec:stegr}.

When the tetrad postulate holds, the global field strengths are easily related to the local ones,
\begin{equation}
R^\rho{}_{\sigma\mu\nu}=e_a{}^\rho e^b{}_\sigma\,\Omega_{\mu\nu}{}^a{}_b,\qquad
T^\rho{}_{\mu\nu}=e_a{}^\rho\,\Theta_{\nu\mu}{}^a,\qquad
Q_{\rho\mu\nu}=e^a{}_\mu e^b{}_\nu\,\Xi_{\rho ab}.
\label{eq:lg-localglobal}
\end{equation}
The affine connection is fixed by the metric jet, torsion, and nonmetricity. With the Schouten brace
\begin{equation}
Q_{\{\mu\alpha\beta\}}:=Q_{\alpha\mu\beta}+Q_{\beta\mu\alpha}-Q_{\mu\alpha\beta},
\label{eq:lg-schouten}
\end{equation}
one finds
\begin{equation}
\Gamma_{\mu\alpha\beta}=\bar\Gamma_{\mu\alpha\beta}+N_{\mu\alpha\beta}
=\tfrac12\bigl(\partial_{\{\mu}g_{\alpha\beta\}}-T_{\{\mu\alpha\beta\}}-Q_{\{\mu\alpha\beta\}}\bigr),
\label{eq:lg-gamma}
\end{equation}
with the distortion $N^\rho{}_{\mu\nu}=-K^\rho{}_{\mu\nu}-L^\rho{}_{\mu\nu}$ split into contortion and disformation as in Sec.~\ref{ssec:conventions}. 

The curvature decomposes into its Levi-Civita part plus distortion terms,
\begin{equation}
R^\rho{}_{\sigma\mu\nu}=\bar R^\rho{}_{\sigma\mu\nu}
+2\,\bar\nabla_{[\mu}N^\rho{}_{|\sigma|\nu]}
+2\,N^\rho{}_{\lambda[\mu}N^\lambda{}_{|\sigma|\nu]},
\label{eq:lg-riemann-decomp}
\end{equation}
with Ricci scalar
\begin{equation}
R=\bar R+g^{\sigma\nu}\Bigl(2\,\bar\nabla_{[\mu}N^\mu{}_{|\sigma|\nu]}+2\,N^\mu{}_{\lambda[\mu}N^\lambda{}_{|\sigma|\nu]}\Bigr).
\label{eq:lg-ricci-decomp}
\end{equation}
Teleparallelism imposes $R^\rho{}_{\sigma\mu\nu}=0$, so Eq.~\eqref{eq:lg-ricci-decomp} solves the pseudo-Riemannian Ricci scalar $\bar R$ in terms of the distortion. 

Metric teleparallelism has $Q_{\rho\mu\nu}=0$, leaving torsion only, and
\begin{equation}
\bar R=2\,\bar\nabla_\mu T^\mu+\mathbb T,
\qquad
\mathbb T=-\tfrac14 T_{\rho\mu\nu}T^{\rho\mu\nu}+\tfrac12 T_{\rho\mu\nu}T^{\mu\nu\rho}+T_\mu T^\mu .
\label{eq:lg-Tscalar}
\end{equation}
Symmetric teleparallelism has a symmetric connection, $T^\rho{}_{\mu\nu}=0$, leaving nonmetricity only, and
\begin{equation}
\bar R=\bar\nabla_\alpha\bigl(\widetilde Q^\alpha-Q^\alpha\bigr)+\mathbb Q,
\qquad
\mathbb Q=-\tfrac14 Q_{\rho\mu\nu}Q^{\rho\mu\nu}+\tfrac12 Q_{\rho\mu\nu}Q^{\mu\nu\rho}+\tfrac14 Q_\mu Q^\mu-\tfrac12 Q_\mu\widetilde Q^\mu,
\label{eq:lg-Qscalar}
\end{equation}
with the traces $Q_\rho=Q_{\rho\lambda}{}^\lambda$ and $\widetilde Q_\mu=Q_{\lambda\mu}{}^\lambda$ of Sec.~\ref{ssec:conventions}. The scalars $\mathbb T$ and $\mathbb Q$ are quadratic in their field strengths. Up to the total derivatives in Eqs.~\eqref{eq:lg-Tscalar} and \eqref{eq:lg-Qscalar}, they are the MTEGR and STEGR Lagrangians. Quadratic dependence permits a field-strength Legendre map. The explicit STEGR and MTEGR Hessians below are nondegenerate on their declared field-strength fibers.

Between the DDW polyvelocity $\partial_\mu\phi^A$ and a general field strength sits the covariant polyvelocity $\nabla_\mu\phi^A$, which Utiyama's treatment of gauge-invariant interactions already singles out \cite{Utiyama}.
An Ehresmann connection on $\pi:Y\to X$ turns the jet bundle $J^1Y$ into the covariant polyvelocity bundle $T^*X\otimes_Y VY$, the fiberwise tensor product over $Y$ of the pulled-back cotangent bundle with the vertical bundle $VY=\ker T\pi$ \cite{Sardanashvily:1994fg}.
Nonmetricity $Q_{\rho\mu\nu}=\nabla_\rho g_{\mu\nu}$ is at once the covariant polyvelocity and the field strength of the metric, so the covariant momentum and the momentum field strength coincide in STEGR. Torsion and curvature are two-forms instead of covariant derivatives of their potentials, so the general field-strength bundle of Sec.~\ref{ssec:fsform} is required for them. This distinction returns in the MTEGR construction of Sec.~\ref{ssec:mtegr}.

The flat connection carries no local dynamics. On a contractible patch, flatness and torsionlessness are solved by four St\"uckelberg functions $\zeta^\alpha$,
\begin{equation}
\Gamma^\rho{}_{\mu\nu}=\frac{\partial x^\rho}{\partial\zeta^\alpha}\,\partial_\mu\partial_\nu\zeta^\alpha,
\label{eq:stueckelberg}
\end{equation}
inertial gauge data instead of propagating fields. The field-strength formulation carries this connection covariantly, without gauge-fixing it away. The covariant bulk Hamilton equations and their Einstein-form identities hold for any fixed flat torsionless $\Gamma$, with the covariant derivative absorbing the inertial connection. The later reductions require their stated additional hypotheses. In particular, the coincident gauge $\Gamma^\rho{}_{\mu\nu}=0$ (i.e.\ $\zeta^\alpha=x^\alpha$, where $Q_{\rho\mu\nu}=\partial_\rho g_{\mu\nu}$) enters in the compact canonical reduction. Globally, it exists only when the flat connection has trivial holonomy.

The local presentation separates the metric-affine target bundle into
three field-strength summands. The local field strengths
$(\Xi_{\mu ab},\Theta_{\mu\nu}{}^a,\Omega_{\mu\nu}{}^a{}_b)$ take
values in
\begin{align}
\mathcal M_Q=T^*X\otimes\Sym^2(\mathbb R^{1,3}),\qquad
\mathcal M_T&=\Lambda^2T^*X\otimes\mathbb R^{1,3},\qquad
\mathcal M_R=\Lambda^2T^*X\otimes\mathfrak g,
\nonumber \\
\mathcal M_{\rm MAG}&=\mathcal M_Q\oplus\mathcal M_T\oplus\mathcal M_R,
\label{eq:lg-trinity}
\end{align}
the local realization of $E_{\rm MAG}$ from Sec.~\ref{ssec:fsform}. The gauge algebra is $\mathfrak g=\mathfrak{gl}(4,\mathbb R)$ for the local connection $\omega_\mu{}^a{}_b$. The local metric $g_{ab}$ is nontrivial exactly when $\mathfrak g$ contains generators outside $\mathfrak{so}(1,3)$. The frame field $e^a{}_\mu$ is tied to the translations of the affine group. The local teleparallel theories gauge the translations $T_4$ for MTEGR and the shear and scale transformations in $GL(4,\mathbb R)/SO(1,3)$ for STEGR \cite{HEHL19951,obukhov2003metric}. The most general gauge group for metric-affine gravity is therefore the 20-parameter affine linear group $AL(4,\mathbb{R})$. A section over $X$ carries the coordinates
\begin{equation}
x^\mu\;\longmapsto\;\bigl(x^\mu,\;g_{ab},\;e^a{}_\mu,\;\omega_\mu{}^a{}_b,\;\Xi_{\mu ab},\;\Theta_{\mu\nu}{}^a,\;\Omega_{\mu\nu}{}^a{}_b\bigr).
\label{eq:lg-coords}
\end{equation}
The momentum field strengths conjugate to the local trinity are Hehl and Obukhov's excitations \cite{Hehl2003,Puetzfeld:2014qba,Obukhov:2018bmf} valued in the dual bundle $\mathcal M^*_{\rm MAG}$,
\begin{equation}
\pi^{\rho ab}=\frac{\partial\sqrt{-g}\,\Lag}{\partial\Xi_{\rho ab}},\qquad
\rho^{\mu\nu}{}_a=\frac{\partial\sqrt{-g}\,\Lag}{\partial\Theta_{\mu\nu}{}^a},\qquad
\sigma^{\mu\nu}{}_a{}^b=\frac{\partial\sqrt{-g}\,\Lag}{\partial\Omega_{\mu\nu}{}^a{}_b}. 
\label{eq:lg-excitations}
\end{equation}
Note that $\sqrt{-g}=|e|\sqrt{-\det(g_{ab})}$. Variational calculus gives the metric-affine Euler--Lagrange equations
\begin{align}
D_\rho\pi^{\rho ab}&=\frac{\partial\sqrt{-g}\,\Lag}{\partial g_{ab}}, \label{eq:lg-local-el-g}\\
D_\mu\rho^{\mu\nu}{}_a&=\frac12\frac{\partial\sqrt{-g}\,\Lag}{\partial e^a{}_\nu}, \label{eq:lg-local-el-e}\\
D_\mu\sigma^{\mu\nu}{}_a{}^b+\pi^{\nu bc}g_{ac}-\rho^{\nu\mu}{}_a\,e^b{}_\mu&=\frac12\frac{\partial\sqrt{-g}\,\Lag}{\partial\omega_\nu{}^a{}_b}. \label{eq:lg-local-el-w}
\end{align}
The fiberwise Legendre transform of Sec.~\ref{ssec:fsform} then defines the local metric-affine Hamiltonian density
\begin{equation}
\sqrt{-g}\,\Ham_{\rm MAG}
=\Xi_{\rho ab}\,\pi^{\rho ab}+\Theta_{\mu\nu}{}^a\,\rho^{\mu\nu}{}_a+\Omega_{\mu\nu}{}^a{}_b\,\sigma^{\mu\nu}{}_a{}^b-\sqrt{-g}\,\Lag,
\label{eq:lg-mag-ham}
\end{equation}
whose Hamilton equations split into three kinematical and three dynamical sets,
\begin{align}
\Xi_{\mu ab}=\frac{\partial(\sqrt{-g}\Ham)}{\partial\pi^{\mu ab}},&\qquad
D_\mu\pi^{\mu ab}=-\frac{\partial(\sqrt{-g}\Ham)}{\partial g_{ab}}, \label{eq:lg-local-ham-g}\\
\Theta_{\mu\nu}{}^a=\frac{\partial(\sqrt{-g}\Ham)}{\partial\rho^{\mu\nu}{}_a},&\qquad
2D_\mu\rho^{\mu\nu}{}_a=-\frac{\partial(\sqrt{-g}\Ham)}{\partial e^a{}_\nu}, \label{eq:lg-local-ham-e}\\
\Omega_{\mu\nu}{}^a{}_b=\frac{\partial(\sqrt{-g}\Ham)}{\partial\sigma^{\mu\nu}{}_a{}^b},&\qquad
D_\mu\sigma^{\mu\nu}{}_a{}^b+\pi^{\nu bc}g_{ac}-\rho^{\nu\mu}{}_a e^b{}_\mu
=-\frac12\frac{\partial(\sqrt{-g}\Ham)}{\partial\omega_\nu{}^a{}_b}. \label{eq:lg-local-ham-w}
\end{align}
Here $D_\mu$ is the gauge-covariant divergence of the weight-one excitation density. The factors of two in the torsion and curvature equations follow from their antisymmetric derivative slots. In particular, the connection equation contains the metric and coframe excitation currents. In the teleparallel sectors the curvature equations in Eqs.~\eqref{eq:lg-local-el-w} and \eqref{eq:lg-local-ham-w} carry little content, since the spin connection is not dynamical there. Lagrange multipliers enforcing flatness are needed only in the global formulation, as discussed below.

\subsection{Global metric-affine gravity}
The global presentation has two potentials $\phi^A=(g_{\mu\nu},\Gamma^\rho{}_{\mu\nu})$ and three field strengths $F_\mu{}^A=(Q_{\rho\mu\nu},T^\rho{}_{\mu\nu},R^\rho{}_{\sigma\mu\nu})$. The variation of the Lagrangian density isolates sources and momenta,
\begin{align}
\delta L(g,\Gamma,Q,T,R)
={}&-\frac12\sqrt{-g}\bigl(T^{\mu\nu}\delta g_{\mu\nu}
+\Delta_\rho{}^{\mu\nu}\delta\Gamma^\rho{}_{\mu\nu}\bigr) \nonumber\\
&+\Pi^{\rho\mu\nu}\delta Q_{\rho\mu\nu}
+P_\rho{}^{\mu\nu}\delta T^\rho{}_{\mu\nu} \nonumber\\
&+\Sigma_\rho{}^{\sigma\mu\nu}\delta R^\rho{}_{\sigma\mu\nu},
\label{eq:lg-global-var}
\end{align}
where the stress tensor and hypermomentum are the current-density sources
\begin{equation}
T^{\mu\nu}=-\frac{2}{\sqrt{-g}}\frac{\partial\sqrt{-g}\,\Lag}{\partial g_{\mu\nu}},
\qquad
\Delta_\rho{}^{\mu\nu}=-\frac{2}{\sqrt{-g}}\frac{\partial\sqrt{-g}\,\Lag}{\partial\Gamma^\rho{}_{\mu\nu}},
\label{eq:lg-sources}
\end{equation}
and the momentum field strengths, or excitations, are
\begin{equation}
\Pi^{\rho\mu\nu}=\frac{\partial\sqrt{-g}\,\Lag}{\partial Q_{\rho\mu\nu}},\qquad
P_\rho{}^{\mu\nu}=\frac{\partial\sqrt{-g}\,\Lag}{\partial T^\rho{}_{\mu\nu}},\qquad
\Sigma_\rho{}^{\sigma\mu\nu}=\frac{\partial\sqrt{-g}\,\Lag}{\partial R^\rho{}_{\sigma\mu\nu}}.
\label{eq:lg-global-momenta}
\end{equation}
The curvature momentum $\Sigma_\rho{}^{\sigma\mu\nu}$ is distinguished from a hypersurface $\Sigma$ or $\Sigma_s$ by its indices.

Torsion is not the derivative of a global potential. Isolating $\Gamma$ also replaces the manifestly covariant expression $Q(g,\nabla g)$ by $Q(g,\partial g,\Gamma)$, admitting the non-tensorial connection into the configuration space. The field strengths and their momenta in Eq.~\eqref{eq:lg-global-momenta} nevertheless remain tensors.

With matter, the connection dependence of $\Lag_{\rm matter}(\psi,\nabla\psi,g)$ sources the hypermomentum $\Delta_\rho{}^{\mu\nu}$. The torsion momentum $P_\rho{}^{\mu\nu}$ then shares information with this matter current. Obukhov and Hehl analyze the corresponding hyperfluid sector in Ref.~\cite{Obukhov:2023yti}. We restrict the following analysis to vacuum gravity.

Expanding the field-strength variations in $(\delta g,\delta\Gamma,\delta\partial g,\delta\partial\Gamma)$ and integrating ordinary coordinate derivatives by parts yields the global field-strength Euler--Lagrange equations
\begin{align}
\partial_\rho\Pi^{\rho\mu\nu}
+2\Gamma^{(\mu}{}_{\lambda\rho}\Pi^{\rho|\lambda|\nu)}
-\frac{\partial L}{\partial g_{\mu\nu}}&=0, \label{eq:lg-global-el-g}\\
\partial_\lambda\Sigma_\rho{}^{\mu\lambda\nu}
-\Sigma_\rho{}^{\beta\nu\lambda}\Gamma^\mu{}_{\beta\lambda}
-\Sigma_\alpha{}^{\mu\lambda\nu}\Gamma^\alpha{}_{\rho\lambda}
+\Pi^{\nu\mu\beta}g_{\rho\beta}-P_\rho{}^{\mu\nu}
-\frac12\frac{\partial L}{\partial\Gamma^\rho{}_{\mu\nu}}&=0,
\label{eq:lg-global-el-G}
\end{align}
valid for arbitrary metric-affine theories, including actions of higher polynomial order in the field strengths. These coordinate expressions are unambiguous for the weight-one excitations and the slot convention in Eq.~\eqref{eq:connection-slot-convention}. They are the direct formal adjoints of $\delta Q$ and $\delta R$. In particular, the connection equation contains the nonmetricity current $\Pi^{\nu\mu\beta}g_{\rho\beta}$ as well as the torsion current. The corresponding canonical action
\begin{equation}
S(g,\Gamma,\Pi,P,\Sigma)=\int d^4x\Bigl(Q_{\rho\mu\nu}\Pi^{\rho\mu\nu}+T^\rho{}_{\mu\nu}P_\rho{}^{\mu\nu}+R^\rho{}_{\sigma\mu\nu}\Sigma_\rho{}^{\sigma\mu\nu}-\sqrt{-g}\,\Ham\Bigr)
\label{eq:lg-global-action}
\end{equation}
delivers five sets of Hamilton equations in total. There are three kinematical equations, 
\begin{equation}
Q_{\rho\mu\nu}=\frac{\partial\sqrt{-g}\,\Ham}{\partial\Pi^{\rho\mu\nu}},\qquad
T^\rho{}_{\mu\nu}=\frac{\partial\sqrt{-g}\,\Ham}{\partial P_\rho{}^{\mu\nu}},\qquad
R^\rho{}_{\sigma\mu\nu}=\frac{\partial\sqrt{-g}\,\Ham}{\partial\Sigma_\rho{}^{\sigma\mu\nu}},
\label{eq:lg-global-ham-kin}
\end{equation}
and two dynamical equations, 
\begin{align}
\partial_\rho\Pi^{\rho\mu\nu}
+2\Gamma^{(\mu}{}_{\lambda\rho}\Pi^{\rho|\lambda|\nu)}
&=-\frac{\partial(\sqrt{-g}\Ham)}{\partial g_{\mu\nu}}, \label{eq:lg-global-ham-g}\\
\partial_\lambda\Sigma_\rho{}^{\mu\lambda\nu}
-\Sigma_\rho{}^{\beta\nu\lambda}\Gamma^\mu{}_{\beta\lambda}
-\Sigma_\alpha{}^{\mu\lambda\nu}\Gamma^\alpha{}_{\rho\lambda}
+\Pi^{\nu\mu\beta}g_{\rho\beta}-P_\rho{}^{\mu\nu}
&=-\frac12\frac{\partial(\sqrt{-g}\Ham)}{\partial\Gamma^\rho{}_{\mu\nu}}. \label{eq:lg-global-ham-G}
\end{align}
The antisymmetric part of Eq.~\eqref{eq:lg-global-ham-G} in $[\mu\nu]$ contains the torsion momentum together with antisymmetric parts of the curvature-excitation divergence and nonmetricity current. It isolates torsion only when those other contributions vanish or have been fixed separately.

One caveat governs the global teleparallel subsectors. Restricting global metric-affine geometry to metric teleparallelism, say, requires Lagrange multipliers enforcing vanishing curvature. Without them the correct equations of motion are not obtained. This traces back to torsion not being the derivative of a global potential. The local formulation avoids the complication, since the frame field acts as a conventional potential for torsion, and it is the route taken for MTEGR in Sec.~\ref{ssec:mtegr}.
The Hamiltonian formulation of MTEGR has been explored extensively by Maluf and collaborators \cite{Maluf:1994ji,Maluf:2000zt}.

\section{Gauge-covariant Hamiltonians for the teleparallel equivalents}
\label{sec:teleparallel}

\subsection{STEGR field equations}
\label{ssec:stegr}

The bulk STEGR action obtained after dropping the divergence in Eq.~\eqref{eq:stegr-action} is a single quadratic invariant of nonmetricity. In terms of the superpotential $\mathfrak P^{\rho\mu\nu}$ defined in Eq.~\eqref{eq:stegr-superpotential},
\begin{equation}
S_{\textrm{STEGR}}\approx-\frac{1}{\kappa^2}\int d^4x\,\sqrt{-g}\;Q_{\rho\mu\nu}\mathfrak P^{\rho\mu\nu}.
\label{eq:stegr-QP}
\end{equation}

The nonmetricity superpotential in Eq.~\eqref{eq:stegr-superpotential} carries the conserved current and the quasilocal energy in symmetric teleparallelism \cite{BeltranJimenez:2019tjy}. The field-strength formulation gives it a second role. Treating nonmetricity as the generalized velocity of the metric, the momentum field strength conjugate to $Q$ is, up to the density factor, the superpotential itself,
\begin{equation}
\Pi^{\rho\mu\nu}=\frac{\partial\sqrt{-g}\,\Lag}{\partial Q_{\rho\mu\nu}}=k\,\mathfrak P^{\rho\mu\nu},\qquad k:=-\frac{2\sqrt{-g}}{\kappa^2}.
\label{eq:stegr-momentum}
\end{equation}
The polymomentum conjugate to nonmetricity is the teleparallel superpotential, the gravitational excitation of the dictionary of Table~\ref{tab:DDW_MAG}.

Recovering nonmetricity from its momentum is the inversion of Eq.~\eqref{eq:stegr-momentum}, and its two independent traces do most of the work. Contracting,
\begin{equation}
\Pi^\rho=g_{\mu\nu}\Pi^{\rho\mu\nu}=k\bigl(Q^\rho-\widetilde Q^\rho\bigr),
\qquad
\widetilde\Pi^\mu=g_{\rho\nu}\Pi^{\rho\mu\nu}=k\Bigl(-\tfrac12\widetilde Q^\mu-\tfrac14 Q^\mu\Bigr),
\label{eq:app-traces}
\end{equation}
an invertible $2\times2$ system for the trace vectors $(Q^\rho,\widetilde Q^\rho)$. Solving it and combining with the cyclic sum $\Pi_{\mu\nu\rho}+\Pi_{\nu\mu\rho}$ returns the full nonmetricity,
\begin{equation}
Q_{\rho\mu\nu}=k^{-1}\Bigl[\Pi_{\mu\nu\rho}+\Pi_{\nu\mu\rho}+\tfrac13\bigl(g_{\mu\nu}(\Pi_\rho-2\widetilde\Pi_\rho)-g_{\rho\mu}(\Pi_\nu+\widetilde\Pi_\nu)-g_{\rho\nu}(\Pi_\mu+\widetilde\Pi_\mu)\bigr)\Bigr].
\label{eq:C-action}
\end{equation}
That this inversion is unobstructed on the whole nonmetricity fiber is the one algebraic fact the construction rests on.

\begin{theorem}[invertibility of the STEGR polysupermetric]
\label{thm:inverse}
Write the momentum map in Eq.~\eqref{eq:stegr-momentum} as $\Pi=k\,\mathcal B Q$ and its inverse in Eq.~\eqref{eq:C-action} as $Q=k^{-1}\mathcal C\,\Pi$. On the $40$-dimensional field-strength fiber $V=T^*M\otimes\Sym^2T^*M$, $\mathcal C\mathcal B=\mathcal B\mathcal C=\id_V$. After using the metric musical map to identify the Hessian map $V\to V^*$ with an endomorphism of $V$, the double-trace-free part of $\mathcal B$ has eigenvalues $\tfrac12$ and $-1$, while the trace part has a $2\times2$ trace-mixing matrix of determinant $-\tfrac34$. The field-strength Hessian has no zero modes, so no gauge fixing is needed to invert this map. At a Lorentzian metric, $\operatorname{Hess}\mathbb Q$ has inertia $(22_+,18_-)$ and the physical Hessian $M=k\mathcal B$ has inertia $(18_+,22_-)$. With unrestricted indices the two inverse contractions equal the symmetric-pair projector
\begin{equation}
M_{\rho\mu\nu|\sigma\alpha\beta}\,
\bigl(M^{-1}\bigr)^{\sigma\alpha\beta|\tau\gamma\delta}
=\delta_\rho^\tau\,\delta^{\gamma}_{(\mu}\delta^{\delta}_{\nu)},
\qquad M=k\mathcal B,\quad M^{-1}=k^{-1}\mathcal C.
\label{eq:projector}
\end{equation}
\end{theorem}

The proof is the block decomposition in App.~\ref{app:inverse}. This fiberwise Legendre regularity is logically separate from gauge redundancy and from the Dirac question of which components propagate, taken up in Sec.~\ref{ssec:legendre}. The trace symmetries of $Q$ do not force zero modes. They fix the domain and the symmetric-pair projector in Eq.~\eqref{eq:projector}.

The Hamiltonian density is now derived, not posited. The theory is purely kinetic, $\sqrt{-g}\,\Lag=\tfrac12 Q_{\rho\mu\nu}\Pi^{\rho\mu\nu}$, so the Legendre transform $\sqrt{-g}\,\Ham=\Pi^{\rho\mu\nu}Q_{\rho\mu\nu}-\sqrt{-g}\,\Lag$ equals $\sqrt{-g}\,\Lag$ on shell. Substituting the inverse in Eq.~\eqref{eq:C-action} expresses it in the momenta,
\begin{equation}
\sqrt{-g}\,\Ham_{\textrm{STEGR}}
=-\frac{\kappa^2}{2\sqrt{-g}}\left(\Pi^{\rho\mu\nu}\Pi_{\mu\nu\rho}+\tfrac16\Pi^\mu\Pi_\mu-\tfrac23\Pi^\mu\widetilde\Pi_\mu-\tfrac13\widetilde\Pi^\mu\widetilde\Pi_\mu\right),
\label{eq:stegr-ham}
\end{equation}
quadratic in the momentum field strength with the inverse polysupermetric as its coefficient. With matter or a cosmological constant the corresponding potential enters with the opposite sign (Sec.~\ref{ssec:legendre}). In vacuum Eq.~\eqref{eq:stegr-ham} is homogeneous of degree two in $\Pi$.

The field-strength Hamilton equations in Eqs.~\eqref{eq:fs-hamilton} are now concrete. The first returns nonmetricity directly from the phase-space variables,
\begin{equation}
Q_{\rho\mu\nu}=\frac{\partial(\sqrt{-g}\,\Ham_{\textrm{STEGR}})}{\partial\Pi^{\rho\mu\nu}}
=-\frac{\kappa^2}{2\sqrt{-g}}(\mathcal C\Pi)_{\rho\mu\nu},
\label{eq:stegr-first-hamilton}
\end{equation}
where $\mathcal C\Pi$ is the bracket displayed explicitly in Eq.~\eqref{eq:C-action}. The second Hamilton equation is
\begin{equation}
\nabla_\rho\Pi^{\rho\mu\nu}=-\frac{\partial(\sqrt{-g}\,\Ham_{\textrm{STEGR}})}{\partial g_{\mu\nu}}\bigg|_\Pi+\frac{\partial(\sqrt{-g}\,\Lag_{\rm m})}{\partial g_{\mu\nu}},
\label{eq:stegr-eom}
\end{equation}
with $\nabla$ the flat torsionless inertial connection. Thus the equation of motion is expressed entirely in the canonical pair $(g_{\mu\nu},\Pi^{\rho\mu\nu})$. For the weight-one momentum density the connection-trace terms cancel, leaving
\begin{equation}
\nabla_\rho\Pi^{\rho\mu\nu}
=\partial_\rho\Pi^{\rho\mu\nu}+\Gamma^\mu{}_{\lambda\rho}\Pi^{\rho\lambda\nu}+\Gamma^\nu{}_{\lambda\rho}\Pi^{\rho\mu\lambda},
\label{eq:divPi}
\end{equation}
so the coincident-gauge form looks connection-free while the covariant statement carries the two surviving index terms.

\begin{samepage}
The fixed-momentum derivative on the right side of Eq.~\eqref{eq:stegr-eom} can be followed directly from Eq.~\eqref{eq:stegr-ham}. Write $\delta_\Pi$ for a metric variation with the contravariant components $\Pi^{\rho\mu\nu}$ held fixed. The required elementary variations are
\begin{align}
\delta_\Pi\!\left(\frac{1}{\sqrt{-g}}\right)
&=-\frac{1}{2\sqrt{-g}}g^{\alpha\beta}\delta g_{\alpha\beta},\\
\delta_\Pi\Pi^\rho
&=\Pi^{\rho\alpha\beta}\delta g_{\alpha\beta},
\qquad
\delta_\Pi\widetilde\Pi^\mu
=\Pi^{(\alpha|\mu|\beta)}\delta g_{\alpha\beta},\\
\delta_\Pi\Pi_{\alpha\beta\rho}
&=\Pi^\lambda{}_{\beta\rho}\delta g_{\alpha\lambda}
+\Pi_\alpha{}^\lambda{}_\rho\delta g_{\beta\lambda}
+\Pi_{\alpha\beta}{}^\lambda\delta g_{\rho\lambda}.
\label{eq:fixed-Pi-variations}
\end{align}
\end{samepage}
At fixed $\Pi$, the metric enters through the density factor, the lowered momentum slots, and the trace contractions. Applying Eq.~\eqref{eq:fixed-Pi-variations} to these terms defines the symmetric tensor $\mathcal V^{\mu\nu}$.

\begin{samepage}
Define
\begin{align}
\mathcal V^{\mu\nu}[\Pi,g]
={}&\tfrac14\Pi^{\alpha\mu\beta}\Pi_\beta{}^\nu{}_\alpha
+\tfrac12\Pi^{(\mu|\alpha\beta|}\Pi_\alpha{}^{\nu)}{}_\beta
+\tfrac1{24}\Pi^\mu\Pi^\nu-\tfrac1{12}\widetilde\Pi^\mu\widetilde\Pi^\nu
-\tfrac16\Pi^{(\mu}\widetilde\Pi^{\nu)} \nonumber\\
&+\tfrac1{12}\Pi_\alpha\Pi^{\alpha\mu\nu}
-\tfrac16\widetilde\Pi_\alpha\Pi^{\alpha\mu\nu}
-\tfrac16\Pi_\alpha\Pi^{(\mu\nu)\alpha}
-\tfrac16\widetilde\Pi_\alpha\Pi^{(\mu\nu)\alpha} \nonumber\\
&+g^{\mu\nu}\left(-\tfrac18\Pi^{\rho\alpha\beta}\Pi_{\alpha\beta\rho}
-\tfrac1{48}\Pi^\rho\Pi_\rho+\tfrac1{24}\widetilde\Pi^\rho\widetilde\Pi_\rho
+\tfrac1{12}\Pi^\rho\widetilde\Pi_\rho\right).
\label{eq:stegr-VPi}
\end{align}
\end{samepage}
The direct fixed-$\Pi$ differentiation is therefore
\begin{equation}
\frac{\partial(\sqrt{-g}\,\Ham_{\textrm{STEGR}})}{\partial g_{\mu\nu}}\bigg|_\Pi
=-\frac{2\kappa^2}{\sqrt{-g}}\,\mathcal V^{\mu\nu}[\Pi,g].
\label{eq:app-dHdg}
\end{equation}
The pure-trace terms in the last line of Eq.~\eqref{eq:stegr-VPi} come from the first variation in Eq.~\eqref{eq:fixed-Pi-variations}. They combine to $-\tfrac18g^{\mu\nu}$ times the Hamiltonian bilinear. 

Substitution of the first Hamilton equation into the fixed-$\Pi$ result expresses the metric derivative in nonmetricity variables without invoking a Lagrangian metric variation. After lowering the free indices, define
\begin{equation}
\mathcal H_{\mu\nu}[Q]
:=-\frac{\kappa^2}{2\sqrt{-g}}g_{\mu\alpha}g_{\nu\beta}
\frac{\partial(\sqrt{-g}\,\Ham_{\textrm{STEGR}})}{\partial g_{\alpha\beta}}\bigg|_\Pi
=\frac{\kappa^4}{-g}g_{\mu\alpha}g_{\nu\beta}\mathcal V^{\alpha\beta}[\Pi,g].
\label{eq:stegr-HQ-def}
\end{equation}
\begin{samepage}
Substitution of $\Pi=k\mathfrak P[Q]$ into every term of Eq.~\eqref{eq:stegr-VPi} gives, directly,
\begin{align}
\mathcal H_{\mu\nu}[Q]={}&-\tfrac12\mathbb Qg_{\mu\nu}
-\tfrac14Q_{\mu\alpha\beta}Q_\nu{}^{\alpha\beta}
-\tfrac12Q_{\alpha\mu\beta}Q^\alpha{}_{\nu}{}^\beta
+\tfrac12Q_{\alpha\mu\beta}Q^\beta{}_{\nu}{}^\alpha \nonumber\\
&+Q_{(\mu|\alpha\beta|}Q^\alpha{}_{\nu)}{}^\beta
+\tfrac14Q_\mu Q_\nu-\tfrac12Q_{(\mu}\widetilde Q_{\nu)} \nonumber\\
&+\tfrac12Q^\alpha Q_{\alpha\mu\nu}-\tfrac12\widetilde Q^\alpha Q_{\alpha\mu\nu}
-\tfrac12Q^\alpha Q_{(\mu\nu)\alpha}.
\label{eq:stegr-HQ}
\end{align}
\end{samepage}
Equation~\eqref{eq:stegr-HQ} is the fixed-momentum Hamiltonian metric derivative in nonmetricity variables. It was obtained by inserting the first Hamilton equation into Eq.~\eqref{eq:stegr-VPi}, not by replacing the Hamiltonian derivative with a Lagrangian variation.

The divergence in the second Hamilton equation requires a separate index calculation. Since $\nabla_\rho g_{\mu\nu}=Q_{\rho\mu\nu}$, lowering its free indices gives
\begin{align}
\mathcal D_{\mu\nu}[Q]
:={}&-\frac{\kappa^2}{2\sqrt{-g}}g_{\mu\alpha}g_{\nu\beta}
\nabla_\rho\Pi^{\rho\alpha\beta} \nonumber\\
={}&\frac1{\sqrt{-g}}\nabla_\rho\!\left(\sqrt{-g}\,\mathfrak P^\rho{}_{\mu\nu}\right)
-\left(\mathfrak P^{\rho\alpha}{}_{\nu}Q_{\rho\alpha\mu}
+\mathfrak P^{\rho\alpha}{}_{\mu}Q_{\rho\alpha\nu}\right) \nonumber\\
={}&\frac1{\sqrt{-g}}\nabla_\rho\!\left(\sqrt{-g}\,\mathfrak P^\rho{}_{\mu\nu}\right)
+\mathcal C_{\mu\nu}[Q],
\label{eq:stegr-div-lowered}
\end{align}
where expansion of the normalized momentum gives
\begin{align}
\mathcal C_{\mu\nu}[Q]={}&
Q_{\alpha\mu\beta}Q^\alpha{}_{\nu}{}^\beta
-Q_{\alpha\mu\beta}Q^\beta{}_{\nu}{}^\alpha
-Q_{(\mu|\alpha\beta|}Q^\alpha{}_{\nu)}{}^\beta
+\tfrac12Q_{(\mu}\widetilde Q_{\nu)} \nonumber\\
&-Q^\alpha Q_{\alpha\mu\nu}+\widetilde Q^\alpha Q_{\alpha\mu\nu}
+\tfrac12Q^\alpha Q_{(\mu\nu)\alpha}.
\label{eq:stegr-cross-Q}
\end{align}
The two quadratic tensors now combine without any further variation. Adding Eqs.~\eqref{eq:stegr-HQ} and~\eqref{eq:stegr-cross-Q} cancels the three temporary structures $Q_{(\mu|\alpha\beta|}Q^\alpha{}_{\nu)}{}^\beta$, $Q_{(\mu}\widetilde Q_{\nu)}$, and $Q^\alpha Q_{(\mu\nu)\alpha}$ and leaves
\begin{align}
\mathcal A_{\mu\nu}[Q]
:={}&\mathcal H_{\mu\nu}[Q]+\mathcal C_{\mu\nu}[Q] \nonumber\\
={}&-\tfrac12\mathbb Qg_{\mu\nu}
-\tfrac14Q_{\mu\alpha\beta}Q_\nu{}^{\alpha\beta}
+\tfrac12Q_{\alpha\mu\beta}Q^\alpha{}_{\nu}{}^\beta
-\tfrac12Q_{\alpha\mu\beta}Q^\beta{}_{\nu}{}^\alpha \nonumber\\
&+\tfrac14Q_\mu Q_\nu-\tfrac12Q^\alpha Q_{\alpha\mu\nu}
+\tfrac12\widetilde Q^\alpha Q_{\alpha\mu\nu}.
\label{eq:stegr-AQ}
\end{align}
The lowered vacuum Hamilton equation is
\begin{equation}
0=\mathcal D_{\mu\nu}+\mathcal H_{\mu\nu}
=\frac1{\sqrt{-g}}\nabla_\rho\!\left(\sqrt{-g}\,\mathfrak P^\rho{}_{\mu\nu}\right)
+\mathcal A_{\mu\nu}[Q].
\label{eq:stegr-eom-Q}
\end{equation}

The tensor in Eq.~\eqref{eq:stegr-eom-Q} can be identified from the torsion-free distortion relation $N^\rho{}_{\mu\nu}=-L^\rho{}_{\mu\nu}$,
\begin{equation}
\Gamma^\rho{}_{\mu\nu}=\bar\Gamma^\rho{}_{\mu\nu}-L^\rho{}_{\mu\nu},
\qquad
L^\rho{}_{\mu\nu}=Q_{(\mu}{}^\rho{}_{\nu)}-\tfrac12Q^\rho{}_{\mu\nu}.
\label{eq:stegr-disformation}
\end{equation}
Substitution into the curvature definition gives
\begin{align}
R^\rho{}_{\sigma\mu\nu}(\Gamma)={}&\bar R^\rho{}_{\sigma\mu\nu}
-\bar\nabla_\mu L^\rho{}_{\sigma\nu}+\bar\nabla_\nu L^\rho{}_{\sigma\mu}
+L^\rho{}_{\lambda\mu}L^\lambda{}_{\sigma\nu}
-L^\rho{}_{\lambda\nu}L^\lambda{}_{\sigma\mu}.
\label{eq:stegr-riemann-L}
\end{align}
Flatness sets the left side to zero and therefore solves for the Levi-Civita curvature,
\begin{align}
\bar R^\rho{}_{\sigma\mu\nu}={}&
\bar\nabla_\mu L^\rho{}_{\sigma\nu}-\bar\nabla_\nu L^\rho{}_{\sigma\mu}
-L^\rho{}_{\lambda\mu}L^\lambda{}_{\sigma\nu}
+L^\rho{}_{\lambda\nu}L^\lambda{}_{\sigma\mu}.
\label{eq:stegr-Rbar-L}
\end{align}
Contracting the first and third indices gives the Ricci tensor and scalar,
\begin{align}
\bar R_{\sigma\nu}={}&\bar\nabla_\rho L^\rho{}_{\sigma\nu}
-\bar\nabla_\nu L^\rho{}_{\sigma\rho}
-L^\rho{}_{\lambda\rho}L^\lambda{}_{\sigma\nu}
+L^\rho{}_{\lambda\nu}L^\lambda{}_{\sigma\rho},
\label{eq:stegr-Ricci-L}\\
\bar R={}&g^{\sigma\nu}\!\left(
\bar\nabla_\rho L^\rho{}_{\sigma\nu}-\bar\nabla_\nu L^\rho{}_{\sigma\rho}
-L^\rho{}_{\lambda\rho}L^\lambda{}_{\sigma\nu}
+L^\rho{}_{\lambda\nu}L^\lambda{}_{\sigma\rho}\right).
\label{eq:stegr-scalar-L}
\end{align}

Equations~\eqref{eq:stegr-superpotential} and~\eqref{eq:stegr-disformation} give
\begin{equation}
\mathfrak P^\rho{}_{\mu\nu}
=L^\rho{}_{\mu\nu}
+\tfrac12\bigl(Q^\rho-\widetilde Q^\rho\bigr)g_{\mu\nu}
-\tfrac14\bigl(\delta^\rho_\mu Q_\nu+\delta^\rho_\nu Q_\mu\bigr).
\label{eq:stegr-PfromL}
\end{equation}
An independent calculation of the Levi-Civita Einstein tensor begins with the flat affine geometry. Using $L^\rho{}_{\mu\rho}=\tfrac12Q_\mu$ in Eq.~\eqref{eq:stegr-Ricci-L}, together with the symmetry of the Levi-Civita Ricci tensor, gives
\begin{align}
 \bar R_{\mu\nu}={}&
 \tfrac12\bar\nabla_\rho
 \left(Q_\mu{}^\rho{}_\nu+Q_\nu{}^\rho{}_\mu-Q^\rho{}_{\mu\nu}\right)
 -\tfrac12\bar\nabla_{(\mu}Q_{\nu)}
 -\tfrac14Q_\lambda
 \left(Q_\mu{}^\lambda{}_\nu+Q_\nu{}^\lambda{}_\mu-Q^\lambda{}_{\mu\nu}\right)
 \nonumber\\
 &+\tfrac18
 \left(Q_\lambda{}^\rho{}_\nu+Q_\nu{}^\rho{}_\lambda-Q^\rho{}_{\lambda\nu}\right)
 \left(Q_\mu{}^\lambda{}_\rho+Q_\rho{}^\lambda{}_\mu-Q^\lambda{}_{\mu\rho}\right)
 \nonumber\\
 &+\tfrac18
 \left(Q_\lambda{}^\rho{}_\mu+Q_\mu{}^\rho{}_\lambda-Q^\rho{}_{\lambda\mu}\right)
 \left(Q_\nu{}^\lambda{}_\rho+Q_\rho{}^\lambda{}_\nu-Q^\lambda{}_{\nu\rho}\right).
 \label{eq:stegr-Ricci-Q}
\end{align}
Every term on the right side is now a function of $Q$ and its first derivative. Contracting Eq.~\eqref{eq:stegr-Ricci-Q} gives
\begin{equation}
 \bar R=\bar\nabla_\rho(\widetilde Q^\rho-Q^\rho)+\mathbb Q.
 \label{eq:stegr-scalar-Q}
\end{equation}

Equations~\eqref{eq:stegr-Ricci-Q} and~\eqref{eq:stegr-scalar-Q} give the Einstein tensor before any comparison with the dynamics,
\begin{equation}
 \bar G_{\mu\nu}
 =\bar R_{\mu\nu}[Q]
 -\tfrac12g_{\mu\nu}
 \left[\bar\nabla_\rho(\widetilde Q^\rho-Q^\rho)+\mathbb Q\right],
 \label{eq:stegr-GQ-uncollected}
\end{equation}
Using Eq.~\eqref{eq:stegr-PfromL} and $\Gamma=\bar\Gamma-L$, the derivative terms in Eqs.~\eqref{eq:stegr-Ricci-Q}--\eqref{eq:stegr-GQ-uncollected} combine into the affine density divergence of $\mathfrak P^\rho{}_{\mu\nu}$. The quadratic products reduce to the seven structures collected in Eq.~\eqref{eq:stegr-AQ}. The curvature identity is
\begin{equation}
\bar G_{\mu\nu}
=\frac1{\sqrt{-g}}\nabla_\rho\!\left(\sqrt{-g}\,\mathfrak P^\rho{}_{\mu\nu}\right)
+\mathcal A_{\mu\nu}[Q].
\label{eq:stegr-GQ}
\end{equation}
No action identity or Lagrangian metric variation enters Eqs.~\eqref{eq:stegr-Ricci-Q}--\eqref{eq:stegr-GQ}. The right sides of Eqs.~\eqref{eq:stegr-eom-Q} and~\eqref{eq:stegr-GQ} were obtained independently and are identical term by term. 

Restoring upper indices and the normalization of $\Pi$ gives the following identity on the STEGR Legendre graph $\Pi=\Pi(Q,g)$, without using the second Hamilton equation:
\begin{equation}
\sqrt{-g}\,\bar G^{\mu\nu}
=-\frac{\kappa^2}{2}\nabla_\rho\Pi^{\rho\mu\nu}
+\frac{\kappa^4}{\sqrt{-g}}\,\mathcal V^{\mu\nu}[\Pi,g],
\label{eq:stegr-phase-einstein}
\end{equation}
or, using Eq.~\eqref{eq:app-dHdg},
\begin{equation}
\nabla_\rho\Pi^{\rho\mu\nu}+\frac{\partial(\sqrt{-g}\,\Ham_{\textrm{STEGR}})}{\partial g_{\mu\nu}}\bigg|_\Pi=-\frac{2}{\kappa^2}\,\sqrt{-g}\;\bar G^{\mu\nu}.
\label{eq:stegr-einstein}
\end{equation}
Equation~\eqref{eq:stegr-einstein} identifies the vacuum second Hamilton equation with $\bar G_{\mu\nu}=0$. Including the matter term in Eq.~\eqref{eq:stegr-eom} gives
\begin{equation}
\sqrt{-g}\,\bar G^{\mu\nu}=\frac{\kappa^2}{4}\,\sqrt{-g}\,T^{\mu\nu},\qquad
T^{\mu\nu}:=-\frac{2}{\sqrt{-g}}\frac{\partial(\sqrt{-g}\,\Lag_{\rm m})}{\partial g_{\mu\nu}},
\label{eq:stegr-sourced}
\end{equation}
Here, $T^{\mu\nu}$ is the Hilbert stress tensor. Since $\kappa^2=32\pi G/c^4$, Eq.~\eqref{eq:stegr-sourced} is $\bar G^{\mu\nu}=8\pi G T^{\mu\nu}/c^4$.

The STEGR Legendre Hessian is regular but indefinite. Appendix~\ref{app:inverse} derives its Lorentzian inertia from the irreducible trace-free and trace sectors. This signature is a property of the full $40$-dimensional fiber and is not a statement about the reduced ADM kinetic metric.

\begin{theorem}[STEGR metric-sector Hamiltonian equivalence]
\label{thm:stegr-equivalence}
Fix a flat torsionless inertial connection and impose compactly supported
metric variations, or boundary conditions that remove the metric surface
term. On this sector, stationary points of the bulk STEGR metric action are in
one-to-one correspondence with stationary points $(g,\Pi)$ of
\begin{equation}
 S_H[g,\Pi]=\int d^4x\left[
 \Pi^{\rho\mu\nu}\nabla_\rho g_{\mu\nu}
 -\sqrt{-g}\,\Ham_{\rm STEGR}(g,\Pi)\right].
 \label{eq:stegr-first-order-action}
\end{equation}
The correspondence is the invertible Legendre map
$\Pi=\partial(\sqrt{-g}\Lag)/\partial Q$, with
$Q_{\rho\mu\nu}=\nabla_\rho g_{\mu\nu}$. It does not vary the inertial
connection or classify its multiplier constraints.
\end{theorem}

\begin{proof}
Write $\ell(g,Q)=\sqrt{-g}\Lag_{\rm STEGR}$ and
$h(g,\Pi)=\sqrt{-g}\Ham_{\rm STEGR}=\Pi\mathbin{\cdot}Q-\ell$.
Theorem~\ref{thm:inverse} makes the map $Q\leftrightarrow\Pi$ invertible.
Variation of Eq.~\eqref{eq:stegr-first-order-action} with respect to $\Pi$
therefore gives
\begin{equation}
 \nabla_\rho g_{\mu\nu}
 =\frac{\partial h}{\partial\Pi^{\rho\mu\nu}}
 =Q_{\rho\mu\nu}.
\end{equation}
At fixed $\Pi$, the Legendre-envelope identity is
\begin{equation}
 \left.\delta_g h\right|_\Pi
 =-\left.\delta_g\ell\right|_Q,
 \label{eq:stegr-envelope-core}
\end{equation}
because the terms proportional to $\delta_gQ$ cancel against
$\Pi=\partial\ell/\partial Q$. Varying $g$ in the first-order action and
integrating by parts gives Eq.~\eqref{eq:stegr-eom}. Substitution of the first
Hamilton equation and Eq.~\eqref{eq:stegr-envelope-core} returns the bulk
metric Euler--Lagrange equation. Conversely, every metric solution lifts
uniquely by $\Pi=\partial\ell/\partial Q$ and satisfies both Hamilton
equations. Thus the two solution sets are in bijection in the declared
fixed-connection sector.
\end{proof}

\subsection{MTEGR field equations}
\label{ssec:mtegr}

MTEGR follows the same construction with the frame field as potential and local torsion $T^a{}_{\mu\nu}=2D_{[\nu}e^a{}_{\mu]}$ as generalized velocity, where $D_\nu e^a{}_\mu=\partial_\nu e^a{}_\mu+\omega_\nu{}^a{}_b\,e^b{}_\mu$ and $\omega$ is the flat spin connection, kept arbitrary throughout.
Up to boundary terms, the MTEGR action is known to factorize into the torsion tensor 
\begin{equation}
S_{\textrm{MTEGR}}\approx-\tfrac{1}{\kappa^2}\!\int|e|\,T^\rho{}_{\mu\nu}S_\rho{}^{\mu\nu},
\end{equation}
with the torsion superpotential $S_\rho{}^{\mu\nu}$, 
\begin{equation}
S_\rho{}^{\mu\nu}=-K^{\mu\nu}{}_\rho+T^\mu\delta^\nu_\rho-T^\nu\delta^\mu_\rho 
\,.
\end{equation}
The momentum field strength is proportional to the torsion superpotential:
\begin{equation}
P_a{}^{\mu\nu}=\frac{\partial|e|\Lag}{\partial T^a{}_{\mu\nu}}=-\frac{2|e|}{\kappa^2}\,e_a{}^\rho S_\rho{}^{\mu\nu} =-\frac{2|e|}{\kappa^2}\left(-K^{\mu\nu}{}_a+T^\mu e^\nu_a-T^\nu e^\mu_a \right),
\label{eq:mtegr-momentum}
\end{equation}
in the contortion convention of Sec.~\ref{ssec:conventions}, with $S_\rho{}^{\mu\rho}=2T^\mu$, $T^\rho{}_{\mu\nu}S_\rho{}^{\mu\nu}=2\,\mathbb T$, and traces $P_a=0$, $\widetilde P^\mu=e^a{}_\nu P_a{}^{\mu\nu}=-\tfrac{4|e|}{\kappa^2}T^\mu$. Thus $|e|\Lag=\tfrac12P_a{}^{\mu\nu}T^a{}_{\mu\nu}$ and $|e|\Ham_{\rm MTEGR}=T^a{}_{\mu\nu}P_a{}^{\mu\nu}-|e|\Lag=|e|\Lag$.

With the coordinate momentum $P^{\rho\mu\nu}:=\eta^{ab}e_b{}^\rho P_a{}^{\mu\nu}$, the map inverts in closed form:
\begin{equation}
T_{\rho\mu\nu}=\frac{\kappa^2}{2|e|}\Bigl( P_{\mu\rho\nu}- P_{\nu\rho\mu}
+\tfrac12\,g_{\rho\mu}\widetilde P_\nu-\tfrac12\,g_{\rho\nu}\widetilde P_\mu\Bigr),
\label{eq:mtegr-inverse}
\end{equation}
Here the frame-contracted slot of $P$ occupies the middle position.
Contracting this formula gives
\begin{equation}
 T^\mu[P]= -\frac{\kappa^2}{4|e|}\,\widetilde P^\mu.
 \label{eq:mtegr-inverse-trace}
\end{equation}
Substitution of Eqs.~\eqref{eq:mtegr-inverse} and
\eqref{eq:mtegr-inverse-trace} into the contortion and superpotential gives
\begin{equation}
 S^{\rho\mu\nu}[T(P)]= -\frac{\kappa^2}{2|e|}P^{\rho\mu\nu}.
 \label{eq:mtegr-inverse-check}
\end{equation}
Equation~\eqref{eq:mtegr-inverse-check} is Eq.~\eqref{eq:mtegr-momentum}
with a raised frame-contracted index. The proposed inverse therefore returns
$P$ under the momentum map. Since both maps act between $24$-dimensional
antisymmetric torsion fibers, this composition proves that the Legendre map is
bijective and that Eq.~\eqref{eq:mtegr-inverse} is its inverse.
The purely kinetic Hamiltonian density is the cyclic pairing, the exact
analogue of Eq.~\eqref{eq:stegr-ham},
\begin{equation}
|e|\Ham_{\textrm{MTEGR}}=-\frac{\kappa^2}{2|e|}\Bigl(P^{\rho\mu\nu} P_{\mu\nu\rho}+\tfrac12\,\widetilde P^\mu\widetilde P_\mu\Bigr).
\label{eq:mtegr-ham}
\end{equation}

The Hamilton equations of Eq.~\eqref{eq:mtegr-ham} again exhibit the Einstein tensor. The frame gradient at fixed $P$ needs no compensating $\omega$ term, since $|e|\Ham$ is an ultralocal Lorentz scalar of $(e,P)$, and the covariant divergence carries only the frame index, $D_\nu P_a{}^{\mu\nu}=\partial_\nu P_a{}^{\mu\nu}-\omega_\nu{}^b{}_a P_b{}^{\mu\nu}$, no coordinate-connection term surviving by the antisymmetry of the pair and the same weight cancellation as Eq.~\eqref{eq:divPi}. The frame derivative $\partial_\nu e^a{}_\mu$ enters the torsion in two antisymmetric slots. For every flat spin connection,
\begin{equation}
2D_\nu P_a{}^{\mu\nu}
+\frac{\partial(|e|\,\Ham_{\textrm{MTEGR}})}{\partial e^a{}_\mu}\bigg|_{P}
=-\frac{4}{\kappa^2}\,|e|\;e_a{}^\rho\,G_\rho{}^\mu ,
\label{eq:mtegr-einstein}
\end{equation}
with $G_{\mu\nu}$ the Levi-Civita Einstein tensor of $g_{\mu\nu}=\eta_{ab}e^a{}_\mu e^b{}_\nu$. The factor $4$ is two independent factors of two, the antisymmetric pair in $\pi=2P$ and the chain rule $\partial g_{\sigma\rho}/\partial e^a{}_\mu=\delta^\mu_\sigma e_{a\rho}+\delta^\mu_\rho e_{a\sigma}$. 

The derivation mirrors the STEGR one closely. The Legendre identity in Eq.~\eqref{eq:legendre-id2} converts the fixed-$P$ gradient into the frame Euler--Lagrange expression of $|e|\Lag$. The action display in Eq.~\eqref{eq:mtegr-action} exhibits $|e|\,\mathbb T$ as $|e|\bar R$ minus the total divergence $2\,\partial_\mu(|e|T^\mu)$, whose Euler--Lagrange expression vanishes. The frame variation of $|e|\bar R$ doubles the Palatini one through the chain rule above. In expanded form, with $S_a{}^{\lambda\nu}:=e_a{}^\rho S_\rho{}^{\lambda\nu}$,
\begin{align}
|e|\,e_a{}^\mu G_\mu{}^{\lambda} &= D_\nu\!\left(|e|\,S_a{}^{\lambda\nu}\right)
+|e|\,e_a{}^{\mu}\left[-\tfrac12\,\mathbb T\,\delta^\lambda_\mu +T^\lambda T_\mu-T^{\lambda}{}_{\mu\alpha}T^\alpha \right. \\
& \qquad\qquad \left.-\tfrac12\left(T_\alpha{}^{\lambda}{}_{\beta}T^{\alpha}{}_{\mu}{}^{\beta}
+T_\alpha{}^{\lambda}{}_{\beta}T^{\beta}{}_{\mu}{}^{\alpha}
+T^{\lambda}{}_{\alpha\beta}T^{\alpha}{}_{\mu}{}^{\beta}\right)\right], \nonumber
\label{eq:mtegr-Tform}
\end{align}
The antisymmetric part of the right side cancels, as required by local Lorentz invariance, and the action argument proves the identity on all sixteen mixed components. The frame-volume variation $\partial(1/|e|)/\partial e^a{}_\mu=-e_a{}^\mu/|e|$ contributes $+\tfrac{\kappa^2}{4}\delta^\lambda_\mu|e|\Ham$ to the field equation. For matter coupled through the metric, the resulting source agrees with the STEGR source in Eq.~\eqref{eq:stegr-sourced}. The superpotential form agrees with the MTEGR equations of Ref.~\cite{Capozziello:2022zzh}, with the Einstein tensor written explicitly.

\subsection{On Legendre regularity and primary constraints}
\label{ssec:legendre}

Legendre regularity has a limited consequence. Let $W_{\mu A|\nu B}=\partial^2L/\partial F_\mu{}^A\partial F_\nu{}^B$ be the field-strength Hessian on the actual fiber $V$. The condition
\begin{equation}
\ker W\big|_V=0, \label{eq:limited-claim}
\end{equation}
excludes primary constraints arising solely from degeneracy of the field-strength Legendre map. It does not remove constraints associated with gauge symmetry, teleparallel restrictions, multipliers, or boundary conditions.

For STEGR, $V=T^*M\otimes\Sym^2T^*M$, and Theorem~\ref{thm:inverse} gives $\ker W|_V=0$.
First-class constraints from spacetime diffeomorphisms, internal gauge transformations, the teleparallel restrictions $T=0$, $R=0$ and their multipliers, and boundary conditions all survive. They are detected by degeneracy directions of the presymplectic form, $\iota_{\delta_\epsilon}\Omega_\Sigma\approx0$, or by the moment maps $G_\Sigma[\epsilon]$ of Sec.~\ref{sec:hypersurface}, not by the pointwise rank of $W$.
As in Maxwell theory, the field-strength Legendre map is regular on the antisymmetric two-form fiber for MTEGR, and it is also regular for STEGR. The covariant Hamilton equations therefore reproduce the field equations without introducing Legendre constraints on these field-strength fibers or applying the Dirac stabilization algorithm to derive those field equations. 

Yang--Mills theory makes the distinction directly. Its curvature Legendre map
is invertible on the antisymmetric fiber,
$\Pi_a{}^{\mu\nu}=\tfrac12\sqrt{|g|}\kappa_{ab}F^{b\mu\nu}$, with quadratic
Hamiltonian
$\Ham_{\rm YM}=\kappa^{ab}\Pi_a{}^{\mu\nu}\Pi_{b\mu\nu}/\sqrt{|g|}$. On a
$t=\mathrm{const}$ surface, however, $A^a_0$ remains the multiplier of the
Gauss constraint $\mathcal G_a=D_iE_a{}^i\approx0$. Thus regularity of the
field-strength map does not remove first-class gauge constraints. The STEGR
claim has the same limited meaning.

For STEGR itself, work in the compact, boundaryless, boundary-subtracted, trivial-holonomy coincident-gauge sector with ADM variables $(N,N^i,\gamma_{ij};\pi_N,\pi_i,\pi^{ij})$ \cite{ArnowittDeserMisner1962,DeWitt1967}. The primary constraints are $\pi_N\approx0$, $\pi_i\approx0$. Stabilization gives
\begin{equation}
\mathcal H_\perp=\frac{\kappa^2}{2\sqrt\gamma}\Bigl(\pi_{ij}\pi^{ij}-\tfrac12\pi^2\Bigr)-\frac{2\sqrt\gamma}{\kappa^2}\,{}^{(3)}\!R\;\approx\;0,
\qquad
\mathcal H_i=-2D_j\pi_i{}^j\;\approx\;0.
\label{eq:adm-constraints}
\end{equation}
In this compact reduced sector, all eight constraints are first class. The ADM Legendre transform has coefficient $A=2/\kappa^2$. The kinetic and curvature coefficients are $1/A=\kappa^2/2$ and $A=2/\kappa^2$, respectively. An overall sign inherited from the action convention multiplies the complete constraint and does not change its zero set or algebra. The degree count is
\begin{equation}
N_{\rm dof}=\tfrac12\bigl(N_{\rm can}-2N_{\rm FC}-N_{\rm SC}\bigr)=\tfrac12(20-2\cdot8-0)=2,
\label{eq:dof}
\end{equation}
the GR value \cite{HenneauxTeitelboim1992}. The Hamiltonian constraint remains part of the first-class set. The field-strength reformulation relocates the Legendre bookkeeping. It does not erase the gauge constraint structure. Appendix~\ref{app:dirac} gives the complete compact reduced classification. Keeping the inertial connection, or imposing flatness and vanishing torsion by multipliers, instead requires a separate stabilized constraint matrix that includes the connection, multiplier, reducibility, boundary, and holonomy sectors.

In MTEGR without the time gauge, Maluf and da Rocha-Neto obtain a first-class set containing the Hamiltonian and vector constraints and find two local degrees of freedom \cite{Maluf:2000zt}. Guzm\'an derives the coincident-gauge STEGR $3+1$ Hamiltonian and its lapse--shift primary structure, but does not compute the complete constraint algebra or degree count \cite{Guzman:2023}. With an independent connection, the additional connection variables are constrained rather than locally propagating \cite{CapozzielloSauro}. None of these results supplies the unreduced rank calculation excluded from the present analysis.

The covariant fiber calculation and the reduced canonical calculation answer different questions. The first proves the explicit inverse $Q_{\rho\mu\nu}\leftrightarrow\Pi^{\rho\mu\nu}$ on the $40$-component nonmetricity fiber at fixed flat torsionless $\Gamma$, together with the Einstein-form Hamilton equation. It is an algebraic metric-sector result, not a Dirac count. The second uses ADM equivalence and the Dirac--Teitelboim brackets in Appendix~\ref{app:dirac} to establish eight first-class constraints and two local degrees of freedom in the compact, boundaryless, boundary-subtracted, trivial-holonomy coincident-gauge sector. The regular field-strength Legendre map removes only primary constraints caused by degeneracy of that map. It does not remove the gauge constraints.

The field-strength Legendre transform behaves like the mechanical one.
Splitting $\Lag=K-U$ into a part homogeneous quadratic in the velocities and
a non-derivative part gives
\begin{equation}
\sqrt{-g}\,\Ham=K+U,\qquad \sqrt{-g}\,\Lag=K-U,
\label{eq:legendre-flip}
\end{equation}
so $\sqrt{-g}\,\Ham=\sqrt{-g}\,\Lag$ for a purely kinetic theory. Matter
potentials, mass terms, sources, and a cosmological constant enter $\Ham$ with
the opposite sign to $\Lag$.

The covariant Hamiltonian density is not the canonical energy density. The
field-strength Legendre transform treats all spacetime derivatives on equal
footing, so $\Ham_\FS$ is a local Lorentz scalar that generates the covariant
Hamilton equations. It is not the Noether charge associated with a chosen
time translation. In source-free Maxwell theory, for example,
$\Ham_{\rm EM}=(E^2-B^2)/2$ vanishes for a null electromagnetic wave, although
$T^{00}=(E^2+B^2)/2$ is positive. Conversely, $\Ham_{\rm EM}$ is positive for
a static electric field even though the field does not evolve. The physical
energy is obtained from the stress tensor or, after a hypersurface split,
from the corresponding time-translation generator.

In the teleparallel metric equations, $\Ham_\FS$ has a different role. The Einstein-form identities in Eqs.~\eqref{eq:stegr-einstein} and~\eqref{eq:mtegr-einstein} split the Einstein tensor into the momentum-divergence term and the fixed-momentum metric derivative of the Hamiltonian density. Thus $\Ham_\FS$ enters the covariant metric field equation without being an energy density.

% =====================================================================

\section{From covariant Hamilton equations to hypersurface evolution}
\label{sec:hypersurface}
\label{ssec:hyperform}

The previous sections gave covariant field equations without choosing a time
coordinate. A hypersurface description probes how the data changes when a spatial slice is
displaced. The field-strength Hamiltonian alone does not generate this
displacement. Once a family of embeddings is chosen, the full generator
contains the normal and tangential constraints and, when present, a boundary
term. This section constructs that generator, verifies the standard ADM
deformation algebra in the declared reduced sector, and asks when swept
four-volume can parametrize a chosen flow.

We work on the metric/ADM covariant phase space after reducing the inertial
connection, flatness, torsion, and multiplier sectors. Here ``red'' denotes
the canonical functional Poisson bracket on $(h_{ij},\pi^{ij})$. The remaining
ADM gauge orbits are not quotiented, so it is neither the
Poisson--Gerstenhaber nor the fully physical bracket. No unreduced closure is
claimed.
The presymplectic construction follows the covariant phase-space framework
\cite{Crnkovic1988CPS,Lee:1990nz,Iyer:1994ys}.

Let $X:\bar\Sigma\hookrightarrow M$ be a spacelike embedding with unit normal
$n^\mu$, induced metric $h_{ij}$, and deformation vector
\begin{equation}
 \xi^\mu= Nn^\mu+N^ie_i{}^\mu,
 \qquad N=\eps_n n_\mu\xi^\mu,
 \qquad N^i=e^i{}_\mu\xi^\mu .
 \label{eq:lapse-shift}
\end{equation}
For the weight-one polymomentum, the proper and coordinate surface momenta are
\begin{equation}
 \bar\Pi^\mu{}_A=\Pi^\mu{}_A/\sqrt{|g|},
 \qquad P_A=\eps_n n_\mu\bar\Pi^\mu{}_A,
 \qquad \pi_{\Sigma A}=\sqrt h\,P_A.
 \label{eq:surface-orientation}
\end{equation}
The standard embedding kinematics needed in Appendix~\ref{app:generator} are
\begin{equation}
 \dot h_{ij}=2N\bar K_{ij}+D_iN_j+D_jN_i,
 \qquad
 \dot n_\mu=(-\eps_nD_iN+\bar K_{ij}N^j)e^i{}_\mu .
 \label{eq:kinematics}
\end{equation}

\subsection{Reduced generator and deformation algebra}
\label{ssec:generator}
The reduced Noether current and pulled-back presymplectic form are
\begin{align}
 J^\mu[\xi]&=\Pi^\mu{}_A\mathcal L_\xi\phi^A
 -\xi^\mu\sqrt{-g}\,\Lag, \label{eq:noether-current}\\
 \Omega_\Sigma(\delta_1,\delta_2)
 &=\int d^3\sigma\,
 \delta\pi_{\Sigma A}\wedge\delta\phi^A .
 \label{eq:presymplectic}
\end{align}
Here $\delta_1$ and $\delta_2$ are tangent to the covariant solution space.
For $X_\xi\phi^A=\mathcal L_\xi\phi^A$, direct variation gives
\begin{equation}
 \delta J^\mu[\xi]
 =\omega^\mu(\delta,X_\xi)+\partial_\nu k^{\mu\nu}_\xi(\delta)
 -\xi^\mu E_A\delta\phi^A,
 \quad
 k^{\mu\nu}_\xi=2\xi^{[\nu}\Pi^{\mu]}{}_A\delta\phi^A .
 \label{eq:master-identity}
\end{equation}
Thus, on shell,
\begin{equation}
 \delta G^J_\Sigma[\xi]
 =\Omega_\Sigma(\delta,X_\xi)
 +\mathcal F_{\partial\Sigma,\xi}(\delta).
 \label{eq:integrated-identity}
\end{equation}
The corner flux must vanish or be integrable before $G_\Sigma[\xi]$ is a
Hamiltonian function \cite{Regge:1974zd,Iyer:1994ys,WaldZoupas:1999wa}.
Appendix~\ref{app:generator} gives the complete variation, embedding
corrections, and boundary representative.

\phantomsection\label{ssec:decomposition}
The lapse--shift decomposition is
\begin{equation}
 G_\Sigma[N,N^i]=\int_\Sigma d\Sigma
 (N\mathcal C_\perp+N^i\mathcal C_i)+B_{\partial\Sigma},
 \label{eq:NC-split}
\end{equation}
with raw densities
\begin{equation}
 \mathcal C^{\rm raw}_\perp
 =\Ham+P_A(\mathcal L_n\phi^A-F_\perp{}^A)
 -\bar\Pi^i{}_AF_i{}^A,
 \qquad
 \mathcal C^{\rm raw}_i=P_A\mathcal L_{e_i}\phi^A .
 \label{eq:raw-densities}
\end{equation}
The field-strength Hamiltonian is one term in the normal constraint density.
It is not by itself the hypersurface-deformation generator.

\phantomsection\label{ssec:algebra}
At fixed lapse and shift, the metric dependence of the normal changes the
closure bracket to the Bergmann--Komar bracket
\cite{Bergmann:1972ww,Teitelboim1973,Hojman:1976vp,Blohmann:2010jd,Bojowald:2016hgh}
\begin{equation}
 [\xi,\eta]_{\BK}=[\xi,\eta]_{\rm Lie}-X_\xi\eta+X_\eta\xi .
 \label{eq:BK}
\end{equation}
Its lapse--shift form is
\begin{align}
 N_{12}&=N^i\partial_iM-M^i\partial_iN,\label{eq:DT-lapse}\\
 N^i_{12}&=N^j\partial_jM^i-M^j\partial_jN^i
 -\eps_nh^{ij}(N\partial_jM-M\partial_jN).
 \label{eq:DT-shift}
\end{align}
For an integrable corner flux, the reduced covariant bracket has the
conditional form
\begin{equation}
 [[G_\Sigma[\xi],G_\Sigma[\eta]]]_{\rm cov}
 =G_\Sigma[[\xi,\eta]_{\BK}]
 +\mathcal C_\Sigma[\xi,\eta]+K^{\rm int}_{\partial\Sigma}[\xi,\eta].
 \label{eq:conditional-closure}
\end{equation}
In the compact, boundaryless, boundary-subtracted, trivial-holonomy
coincident-gauge sector, the reduced generators are the ADM generators and
close without extension,
\begin{equation}
 \{G^{\rm STEGR}_\Sigma[N,N^i],G^{\rm STEGR}_\Sigma[M,M^i]\}_{\rm red}
 =G^{\rm STEGR}_\Sigma[N_{12},N^i_{12}]_{\rm red}.
 \label{eq:reduced-closure}
\end{equation}
The corresponding eight first-class constraints and two-degree-of-freedom
count are proved in Appendix~\ref{app:dirac}. The unreduced connection,
multiplier, reducibility, boundary, and holonomy sectors require a separate
stabilized constraint matrix.

\phantomsection\label{ssec:charges}
For nonempty boundaries, the charge and its integrability depend on the chosen
boundary phase space \cite{Regge:1974zd,Brown:1992br,WaldZoupas:1999wa}. Those
details are not needed for the compact reduced classification and are left to
Appendix~\ref{app:generator}.

\subsection{Evolution along a family of embeddings}
\label{ssec:classical-evolution}

The word ``evolution'' has two related meanings in this formulation. The
field-strength Hamiltonian first gives a covariant system of spacetime partial
differential equations. For STEGR, its first Hamilton equation returns every
component of $Q_{\rho\mu\nu}=\nabla_\rho g_{\mu\nu}$ from $(g,\Pi)$, while its
second equation fixes $\nabla_\rho\Pi^{\rho\mu\nu}$. These equations contain no
preferred time coordinate. In this sense, $\Ham_\FS$ generates spacetime
development through the variational equations \eqref{eq:fs-hamilton}. It is
not an ordinary canonical Hamiltonian associated with one foliation.

To obtain a one-parameter evolution, choose a smooth family of spacelike
embeddings $X_s:\bar\Sigma\hookrightarrow M$. Let its deformation vector be
$\xi^\mu=Nn^\mu+N^ie_i{}^\mu$ as in Eq.~\eqref{eq:lapse-shift}. For a
differential form $F$, differentiation of its pullback obeys the kinematical
identity
\begin{equation}
 \frac{d}{ds}X_s^*F=X_s^*(\mathcal L_\xi F).
 \label{eq:pullback-evolution}
\end{equation}
The corresponding formula for hypersurface tensors includes the standard
normal and tangential index projections.
The lapse and shift therefore select a linear combination of the spacetime
changes determined by the covariant field equations. They do not alter those
equations.

On the reduced covariant phase space, the same deformation is Hamiltonian
only when its corner flux vanishes or is integrable. Equation
\eqref{eq:integrated-identity} then defines the corrected generator
$G_\Sigma[N,N^i]$. For a reduced hypersurface functional $\mathcal O_\Sigma$
with no separate explicit dependence on $s$,
\begin{equation}
 \frac{d\mathcal O_\Sigma}{ds}
 =X_\xi[\mathcal O_\Sigma]
 =\{\mathcal O_\Sigma,G_\Sigma[N,N^i]\}_{\rm red}.
 \label{eq:surface-hamilton-evolution}
\end{equation}
An explicit dependence on the embedding adds the corresponding partial
derivative. Equations \eqref{eq:pullback-evolution} and
\eqref{eq:surface-hamilton-evolution} give the link between covariant
Hamilton equations and hypersurface Hamiltonian evolution. The generator is
the full lapse--shift expression \eqref{eq:NC-split}, including its constraint
and boundary terms. The scalar $\Ham_\FS$ is only one term in its raw normal
density.

An arbitrary lapse $N(x)$ assigns an independent infinitesimal normal
displacement to each point of the slice. This is the classical content of
many-fingered evolution \cite{Teitelboim1973,Hojman:1976vp,Isham:1984sb}:
\begin{equation}
 \delta_{(N,N^i)}\mathcal O_\Sigma
 =\{\mathcal O_\Sigma,G_\Sigma[N,N^i]\}_{\rm red}.
 \label{eq:many-fingered-classical}
\end{equation}
Evolution along one prescribed family $X_s$ requires only the generator along
that family. Consistent composition of arbitrary local deformations requires
the closure relations \eqref{eq:DT-lapse}--\eqref{eq:reduced-closure}. Their
commutator is another normal--tangential deformation. After quotienting by
the resulting diffeomorphisms, this gives classical refoliation equivalence in
the declared compact reduced sector. The paper does not extend that statement
to the unreduced connection, multiplier, boundary, or nontrivial-holonomy
sectors.

For a closed universe, the reduced generator is a sum of first-class
constraints and has no non-gauge bulk part. Equations
\eqref{eq:pullback-evolution} and \eqref{eq:surface-hamilton-evolution} still
reconstruct the fields on neighboring embeddings. They relate different
representatives of the same diffeomorphism class. A physical time evolution
then requires a clock choice, boundary time, or relational observable. The
regularity of the field-strength Legendre map does not change this canonical
fact \cite{Kuchar:1991qf,Isham:1992ms}.

\subsection{Swept four-volume and the lapse zero mode}
\label{ssec:volume-time}

Swept four-volume is considered as a scalar path parameter for the Hamiltonian
flow generated by $G_\Sigma[N,N^i]$, not as a new Hamiltonian.

The lapse field contains one spatially averaged mode and infinitely many
shape modes. On a compact slice of finite volume
\begin{equation}
 V_\Sigma=\int_\Sigma d\Sigma,
 \qquad
 N_0=\frac{1}{V_\Sigma}\int_\Sigma N\,d\Sigma,
 \qquad
 N_\perp=N-N_0,
 \label{eq:lapse-zero-mode}
\end{equation}
so $\int_\Sigma N_\perp d\Sigma=0$. The decomposition depends on the induced
measure of the slice. A noncompact slice requires a finite averaging region or
a specified smearing function.

Let $\epsilon$ be the oriented spacetime volume form. Its contraction with the
deformation vector obeys $X_s^*(\iota_\xi\epsilon)=N\,d\Sigma$, since the shift
is tangent to the slice. The oriented four-volume swept out by the family of
slices therefore satisfies
\begin{equation}
 \frac{d\tau}{ds}
 =\int_{\Sigma_s}N(x)\,d\Sigma
 =V_{\Sigma_s}N_0(s).
 \label{eq:swept-volume}
\end{equation}
For future-directed slices with positive lapse, $\tau$ is the ordinary swept
four-volume. With a signed lapse, it is an oriented volume parameter.

Combining Eqs.~\eqref{eq:surface-hamilton-evolution} and
\eqref{eq:swept-volume}, along a chosen path with $V_{\Sigma_s}N_0(s)\neq0$,
gives
\begin{equation}
 \frac{\dd\mathcal O_\Sigma}{\dd\tau}
 =\frac{1}{V_{\Sigma_s}N_0(s)}
 \{\mathcal O_\Sigma,G_\Sigma[N,N^i]\}_{\rm red}.
 \label{eq:volume-reparametrized-hamilton-flow}
\end{equation}
The numerator retains the full lapse, the shift, and any boundary term.
Equation~\eqref{eq:volume-reparametrized-hamilton-flow} therefore
reparametrizes one chosen flow but is not a Hamiltonian law determined by
$\tau$ alone. The modes $N_\perp(x)$ change the local embedding without
changing $\tau$ to first order. An autonomous volume-time law requires a
reduction that removes these shape modes, such as the homogeneous sector
below.

\begin{example}[Diagonal Bianchi I]
\label{ex:bianchi-volume-time}
The homogeneous reduction can be evaluated explicitly. Take a compact
boundaryless coordinate cell of volume $V_0$, coincident gauge, trivial
holonomy, vanishing shift, and
\begin{equation}
 \dd\ell^2=-N(u)^2\dd u^2
 +\sum_{i=1}^3e^{2\beta_i(u)}(\dd x^i)^2,
 \qquad
 V=V_0e^{\beta_1+\beta_2+\beta_3}.
 \label{eq:bianchi-metric-main}
\end{equation}
The nonzero nonmetricity components are
$Q_{u00}=-2N\dot N$ and
$Q_{uii}=2e^{2\beta_i}\dot\beta_i$, with no sum on $i$. Substitution into
the quadratic nonmetricity scalar in Eq.~\eqref{eq:stegr-action} gives
\begin{equation}
 \mathbb Q_{\rm BI}
 =\frac{1}{N^2}\left[
 \sum_i\dot\beta_i^2-\left(\sum_i\dot\beta_i\right)^2
 \right]
 =\frac{1}{N^2}\dot{\boldsymbol\beta}^{\mathsf T}
 A\dot{\boldsymbol\beta},
 \qquad
 A=I-\boldsymbol 1\boldsymbol 1^{\mathsf T}.
 \label{eq:bianchi-Q-main}
\end{equation}
All $\dot N$ terms cancel. The boundary-subtracted STEGR bulk action therefore
reduces to
\begin{equation}
 L_{\rm BI}
 =-\frac{2V}{\kappa^2N}
 \dot{\boldsymbol\beta}^{\mathsf T}A\dot{\boldsymbol\beta},
 \qquad
 \boldsymbol P
 =-\frac{4V}{\kappa^2N}A\dot{\boldsymbol\beta}.
 \label{eq:bianchi-LP-main}
\end{equation}
Since $A^{-1}=I-\tfrac12\boldsymbol 1\boldsymbol 1^{\mathsf T}$, the Legendre
transform is
\begin{equation}
 H_{\rm BI}^{(u)}
 =\boldsymbol P^{\mathsf T}\dot{\boldsymbol\beta}-L_{\rm BI}
 =N\mathcal C_{\rm BI},
 \qquad
 \mathcal C_{\rm BI}
 =-\frac{\kappa^2}{8V}
 \boldsymbol P^{\mathsf T}A^{-1}\boldsymbol P\approx0.
 \label{eq:bianchi-H-main}
\end{equation}
For a phase-space function $F(\boldsymbol\beta,\boldsymbol P)$ with no explicit
$u$ dependence, Hamilton's equation and Eq.~\eqref{eq:swept-volume} give
\begin{equation}
 \frac{\dd F}{\dd u}=N\{F,\mathcal C_{\rm BI}\}_{\rm red},
 \qquad
 \frac{\dd\tau}{\dd u}=NV,
 \qquad
 \frac{\dd F}{\dd\tau}
 =\frac{1}{V}\{F,\mathcal C_{\rm BI}\}_{\rm red}.
 \label{eq:bianchi-volume-flow-main}
\end{equation}
Swept four-volume is therefore an exact reparametrization along this
homogeneous constraint orbit. The calculation does not reconstruct the lapse
shape modes excluded by the truncation.
\end{example}

\section{Conclusions}
\label{sec:conclusions}

Treating torsion and nonmetricity as generalized velocity fields gives
gauge-covariant Hamilton formulations of MTEGR and STEGR. Their momentum
fields are proportional to the corresponding teleparallel superpotentials. In
both theories, the second field-strength Hamilton equation reproduces the
Einstein equation in torsion or nonmetricity variables.
The STEGR inverse supplies the explicit Hamiltonian on all $40$ nonmetricity
components without a coincident-gauge choice. For fixed flat torsionless
connection and the stated boundary conditions,
Theorem~\ref{thm:stegr-equivalence} proves equivalence with the bulk metric
action.

The reduced covariant phase space supplies the hypersurface generator,
corner flux, and conditional Bergmann--Komar algebra.  In the compact,
boundaryless, boundary-subtracted, trivial-holonomy coincident-gauge sector, this reduces to the
standard first-class ADM system. This establishes that the covariant
field-strength equations are compatible with the usual refoliation structure
after reduction to hypersurface data. The nonzero field-strength Hamiltonian is a
term in the raw normal density. It does not by itself generate non-gauge bulk
evolution on a closed universe.
The diagonal Bianchi I example realizes swept volume as an exact
reparametrization of the $H_{\rm BI}=N\mathcal C_{\rm BI}$ flow, with lapse
shape modes excluded.

\begin{appendices}
\numberwithin{equation}{section}

%======================================================================
\numberwithin{table}{section}
\section{Conventions and symbol summary}
\label{app:conventions}

This appendix fixes the notation used for the metric-affine and hypersurface
variables. It also records the derivative operators used in the
field-strength Hamilton equations.

The global metric-affine potentials and field strengths are
\begin{equation}
(g_{\mu\nu},\Gamma^\rho{}_{\mu\nu})
\quad\longleftrightarrow\quad
(Q_{\rho\mu\nu},T^\rho{}_{\mu\nu},R^\rho{}_{\sigma\mu\nu}).
\end{equation}
The local potentials and field strengths are
\begin{equation}
(g_{ab},e^a{}_{\mu},\omega_\mu{}^a{}_b)
\quad\longleftrightarrow\quad
(\Xi_{\mu ab},\Theta_{\mu\nu}{}^a,\Omega_{\mu\nu}{}^a{}_b).
\end{equation}
Their component definitions are
\begin{align}
\Xi_{\mu ab}:=D_\mu g_{ab}
&=\partial_\mu g_{ab}
-\omega_\mu{}^c{}_a g_{cb}
-\omega_\mu{}^c{}_b g_{ac},
\nonumber\\
\Theta_{\mu\nu}{}^a:=2D_{[\mu}e^a{}_{\nu]}
&=2\partial_{[\mu}e^a{}_{\nu]}
+2\omega_{[\mu}{}^a{}_b e^b{}_{\nu]},
\nonumber\\
\Omega_{\mu\nu}{}^a{}_b
&:=2\partial_{[\mu}\omega_{\nu]}{}^a{}_b
+2\omega_{[\mu}{}^a{}_c\omega_{\nu]}{}^c{}_b.
\label{eq:app-local-field-strengths}
\end{align}
With $T^\rho{}_{\mu\nu}=2\Gamma^\rho{}_{[\mu\nu]}$, the tetrad postulate gives
\begin{align}
Q_{\rho\mu\nu}&=e^a{}_{\mu}e^b{}_{\nu}\,\Xi_{\rho ab},
\nonumber\\
T^a{}_{\mu\nu}:=e^a{}_{\rho}T^\rho{}_{\mu\nu}
&=\Theta^a{}_{\nu\mu},
\nonumber\\
R^\rho{}_{\sigma\mu\nu}&=e_a{}^\rho e^b{}_{\sigma}\,
\Omega_{\mu\nu}{}^a{}_b.
\label{eq:app-local-global-fields}
\end{align}

The excitation notation forms two parallel triples. Locally, the momenta
conjugate to $(\Xi,\Theta,\Omega)$ are $(\pi,\rho,\sigma)$ as defined in
Eq.~\eqref{eq:lg-excitations}. Globally, the momenta conjugate to
$(Q,T,R)$ are $(\Pi,P,\Sigma)$ as defined in
Eq.~\eqref{eq:lg-global-momenta}. In the coframe sector the relation
$\Theta^a{}_{\mu\nu}=T^a{}_{\nu\mu}$ gives
$\rho^{\mu\nu}{}_a=P_a{}^{\nu\mu}$ by the chain rule. As such, $\rho$ and $P$
are not distinct torsion momenta. $\rho$ accompanies the local variable
$\Theta$, while $P$ accompanies the global variable $T$. The same local/global
notation change occurs in all three rows of the metric-affine trinity.

For STEGR, the weight-one momentum density and the nondensitized
nonmetricity superpotential are related by
\begin{equation}
\Pi^{\rho\mu\nu}=k\mathfrak P^{\rho\mu\nu},\qquad
k=-\frac{2\sqrt{-g}}{\kappa^2}.
\label{eq:app-momentum-dictionary}
\end{equation}
The indexed tensor $\Sigma_\rho{}^{\sigma\mu\nu}$ is
distinguished from a hypersurface $\Sigma$ by its indices.

The derivative $\partial_\mu$ is a coordinate derivative. The operators
$\nabla_\mu$, $D_\mu$, and $\bar\nabla_\mu$ are the full affine, local-frame,
and Levi--Civita covariant derivatives. For a general affine first-jet map,
the second Hamilton equation contains the gauge-covariant Euler operator
$\mathscr E_{\mathcal F}$ defined in Eq.~\eqref{eq:affine-euler-operator}.
It reduces to $\nabla_\mu\Pi^\mu{}_A$ in the identity-symbol sector, which
includes scalar fields and STEGR nonmetricity. Torsion and curvature retain
the projections produced by their antisymmetrizing symbols.

\begin{table}[htbp]
\centering
\caption{Hypersurface objects used in the deformation construction.}
\label{tab:app-hypersurface-objects}
\begin{tabular}{@{}p{0.28\linewidth}>{\raggedright\arraybackslash}p{0.66\linewidth}@{}}
\toprule
Object & Definition or role \\
\midrule
$X:\bar\Sigma\hookrightarrow M$ & Admissible embedding of the fixed oriented manifold $\bar\Sigma$ into spacetime. \\
$h_{ij}$ & Metric induced on the image hypersurface $\Sigma=X(\bar\Sigma)$. \\
$n_\mu$ & Unit normal, normalized by $n_\mu n^\mu=\eps_n=-1$. \\
$\dd\Sigma$ & Proper hypersurface measure. The directed weight-zero surface element is $\eps_n n_\mu\dd\Sigma$. \\
$\widetilde n_\mu\dd^{,n-1}\sigma$ & Pullback of the coordinate current form, with $\widetilde n_\mu=\eps_n\sqrt h\,n_\mu/\sqrt{|g|}$. \\
$P_A$ & Proper normal momentum $P_A=\eps_n n_\mu\bar\Pi^\mu{}_A$. \\
$\pi_{\Sigma A}$ & Coordinate momentum density $\pi_{\Sigma A}=\widetilde n_\mu\Pi^\mu{}_A=\sqrt h\,P_A$. \\
$\theta^\mu$ & Symplectic-potential current $\theta^\mu=\Pi^\mu{}_A\,\delta\phi^A$. \\
$\Omega_\Sigma$ & Pulled-back presymplectic form defined in Eq.~\eqref{eq:presymplectic}. \\
\bottomrule
\end{tabular}
\end{table}

The local Euler--Lagrange equations are Eqs.~\eqref{eq:lg-local-el-g}--
\eqref{eq:lg-local-el-w}, and their Hamilton form is
Eqs.~\eqref{eq:lg-local-ham-g}--\eqref{eq:lg-local-ham-w}. The factor of two
in the coframe equation comes from the antisymmetric derivative symbol, not
from a second momentum convention. The curvature-sector equation is used
only when the connection is dynamical. The global equations comprise the
three kinematical equations in Eq.~\eqref{eq:lg-global-ham-kin} and the two
dynamical equations in Eqs.~\eqref{eq:lg-global-ham-g}--
\eqref{eq:lg-global-ham-G}. The STEGR metric equation is derived in
Sec.~\ref{ssec:stegr}.

%======================================================================
\section{Invertibility of the polysupermetric}
\label{app:inverse}

This appendix proves that the STEGR Legendre map is regular on the
full nonmetricity fiber. The result is used in the main text to
justify the covariant Hamiltonian density.
\begin{proof}[Proof of Theorem~\ref{thm:inverse}]
Decompose $V=V_0\oplus V_{\rm tr}$ into the double-trace-free and trace parts,
\begin{equation}
Q_{\rho\mu\nu}=\mathring Q_{\rho\mu\nu}+g_{\mu\nu}A_\rho+g_{\rho(\mu}B_{\nu)},
\qquad
A_\rho=\tfrac{5}{18}Q_\rho-\tfrac19\widetilde Q_\rho,
\quad
B_\rho=-\tfrac19 Q_\rho+\tfrac49\widetilde Q_\rho,
\label{eq:app-trace-decomposition}
\end{equation}
so that $\mathring Q_{\rho\lambda}{}^\lambda=\mathring Q_{\lambda\rho}{}^\lambda=0$.

On $V_0$ all trace terms in $\mathcal B$ and $\mathcal C$ drop, so with the index-cycling operator $(\mathsf AX)_{\rho\mu\nu}=X_{\mu\nu\rho}+X_{\nu\mu\rho}$ one has $\mathcal B|_{V_0}=\tfrac12(\mathsf A-\id)$ and $\mathcal C|_{\mathcal B(V_0)}=\mathsf A$. The only identity needed is the polynomial identity $\mathsf A^2=\mathsf A+2\,\id$ on $V$, which follows in two lines from the symmetry $X_{\rho\mu\nu}=X_{\rho\nu\mu}$:
\begin{equation}
(\mathsf A^2X)_{\rho\mu\nu}
=X_{\nu\rho\mu}+X_{\rho\nu\mu}+X_{\mu\rho\nu}+X_{\rho\mu\nu}
=X_{\nu\mu\rho}+X_{\mu\nu\rho}+2X_{\rho\mu\nu}
=(\mathsf AX)_{\rho\mu\nu}+2X_{\rho\mu\nu}.
\end{equation}
Hence, $\mathcal C\mathcal B|_{V_0}=\tfrac12(\mathsf A^2-\mathsf A)=\id_{V_0}$. The eigenspaces have the irreducible decomposition
\begin{equation}
 V_0=\Sym^3_0T^*M\oplus[2,1]_0,
 \qquad
 \dim\Sym^3_0T^*M=\dim[2,1]_0=16.
 \label{eq:app-tracefree-irreps}
\end{equation}
After using the metric musical map to identify the Hessian map with an
endomorphism of $V_0$, the cycling operator $\mathsf A$ has eigenvalue $2$ on
$\Sym^3_0T^*M$ and $-1$ on $[2,1]_0$. Therefore $\mathcal B|_{V_0}$ has
eigenvalue $\tfrac12$ with multiplicity $16$ and $-1$ with multiplicity $16$.
Both are nonzero.

Taking the two traces of $\mathcal B$ in four dimensions, with $p=\mathcal BQ$,
\begin{equation}
\begin{pmatrix}p_\rho\\ \widetilde p_\rho\end{pmatrix}
=\begin{pmatrix}1&-1\\ -\tfrac14&-\tfrac12\end{pmatrix}
\begin{pmatrix}Q_\rho\\ \widetilde Q_\rho\end{pmatrix},
\qquad \det=-\tfrac34\neq0,
\label{eq:app-traceblock}
\end{equation}
with inverse $Q_\rho=\tfrac23 p_\rho-\tfrac43\widetilde p_\rho$, $\widetilde Q_\rho=-\tfrac13 p_\rho-\tfrac43\widetilde p_\rho$, which reproduces the trace inversion in Eq.~\eqref{eq:app-traces}. Eq.~\eqref{eq:app-traceblock} shows the trace sector is an invertible $2\times2$ system without zero modes.

Substituting $p=\mathcal BQ$ into $\mathcal C$, the index-cycling part gives
\begin{equation}
p_{\mu\nu\rho}+p_{\nu\mu\rho}
=Q_{\rho\mu\nu}-\tfrac12 g_{\mu\nu}Q_\rho+g_{\rho\nu}\bigl(\tfrac14 Q_\mu-\tfrac12\widetilde Q_\mu\bigr)+g_{\rho\mu}\bigl(\tfrac14 Q_\nu-\tfrac12\widetilde Q_\nu\bigr),
\end{equation}
while the trace correction in $\mathcal C$, evaluated with Eq.~\eqref{eq:app-traceblock}, gives exactly the negatives of the trace terms above. Adding the two displays cancels every trace term and leaves $(\mathcal C\mathcal BQ)_{\rho\mu\nu}=Q_{\rho\mu\nu}$ for all $Q\in V$. Since $V$ is finite dimensional, the left inverse is two-sided, and written with unrestricted indices both contractions equal the symmetric-pair projector in Eq.~\eqref{eq:projector}. Note that with weighted symmetrization the projector $\delta^\tau_\rho\delta^\gamma_{(\mu}\delta^\delta_{\nu)}$ already contains the factor $\tfrac12$. This proves $\mathcal C\mathcal B=\mathcal B\mathcal C=\id_V$.

The Lorentzian inertia follows from the same decomposition. The natural contraction on
$V=T^*M\otimes\Sym^2T^*M$ has signature $(24,16)$. On the two trace embeddings its internal Gram matrix is
\begin{equation}
 \begin{pmatrix}4&1\\[2pt]1&5/2\end{pmatrix},
 \qquad \det=9>0,
 \label{eq:app-trace-natural-gram}
\end{equation}
times the Lorentzian vector metric. The trace sector therefore has signature $(6,2)$, so $V_0$ has signature $(18,14)$. More explicitly, a symmetric rank-three tensor has positive components with zero or two time indices and negative components with one or three time indices. Removing its vector trace gives
\begin{equation}
 \operatorname{sig}(\Sym^3_0T^*M)=(10,6),
 \qquad
 \operatorname{sig}([2,1]_0)=(8,8).
 \label{eq:app-tracefree-signatures}
\end{equation}
Multiplication by the $\mathcal B$ eigenvalues $\tfrac12$ and $-1$ therefore gives the trace-free quadratic form signature $(18,14)$.

It remains to count the trace block. Substituting
$Q_{\rho\mu\nu}=g_{\mu\nu}A_\rho+g_{\rho(\mu}B_{\nu)}$ with the coefficients in Eq.~\eqref{eq:app-trace-decomposition} into the nonmetricity scalar gives
\begin{equation}
 \mathbb Q_{\rm tr}
 =\frac1{72}\left(
 11Q_\rho Q^\rho-16Q_\rho\widetilde Q^\rho
 -4\widetilde Q_\rho\widetilde Q^\rho\right).
 \label{eq:app-trace-quadratic}
\end{equation}
For each spacetime component, the Hessian on $(Q_\rho,\widetilde Q_\rho)$ is
\begin{equation}
 \begin{pmatrix}
 11/36&-2/9\\[2pt]
 -2/9&-1/9
 \end{pmatrix},
 \qquad
 \det=-\frac1{12}.
 \label{eq:app-trace-hessian}
\end{equation}
The internal trace matrix has one positive and one negative eigenvalue. Tensoring it with a Lorentzian vector therefore gives trace-sector signature $(4,4)$. Hence
\begin{equation}
 \operatorname{sig}(\operatorname{Hess}\mathbb Q)=(22,18).
 \label{eq:app-stegr-inertia}
\end{equation}
The physical Legendre Hessian is $M=k\mathcal B$ with $k=-2\sqrt{-g}/\kappa^2<0$, so its inertia is $(18,22)$. This is the signature of the full field-strength Hessian, not of a spatially reduced kinetic operator.
\end{proof}

%======================================================================
\section{The generator identity, deformation kinematics, and the classical algebra}
\label{app:generator}

This appendix derives the reduced metric/ADM hypersurface-deformation
generator from the covariant phase-space identity and identifies the boundary
and constraint terms that enter its classical Bergmann--Komar algebra. It
does not include the independent inertial connection, teleparallel
multipliers, their momenta, or reducibility data.

Supplementing Eq.~\eqref{eq:kinematics}, when the ambient metric is also varied,
\begin{equation}
\dot h_{ij}\big|_{\dot g}=e_i{}^\mu e_j{}^\nu\dot g_{\mu\nu},
\qquad
\dot n_\mu\big|_{\dot g}=\tfrac{\eps_n}{2}n_\mu n^\alpha n^\beta\dot g_{\alpha\beta},
\qquad
\frac{d}{ds}\dd\Sigma\Big|_{\dot g}=\tfrac12 h^{ij}e_i{}^\mu e_j{}^\nu\dot g_{\mu\nu}\,\dd\Sigma .
\end{equation}
The first of Eq.~\eqref{eq:kinematics} shows the normal's deformation is purely tangential at fixed metric. Holding the spacetime vector $\xi^\mu$ fixed is not the same as holding its descriptors fixed. Next, we derive the master identity in Eq.~\eqref{eq:master-identity}.

On the reduced metric-sector covariant phase space, the first variation of the Lagrangian density defines $\theta^\mu$ via $\delta L_{\rm dens}=E_A\delta\phi^A+\partial_\mu\theta^\mu(\delta)$. Diffeomorphism covariance for field-independent $\xi$ gives $\delta_\xi L_{\rm dens}=\partial_\mu(\xi^\mu L_{\rm dens})$, whence the current in Eq.~\eqref{eq:noether-current} with $\partial_\mu J^\mu[\xi]=-E_AX_\xi\phi^A$. Varying $J^\mu[\xi]$ at fixed $\xi$, using $[\delta,X_\xi]=0$ and Cartan's formula $\mathcal L_\xi\theta=i_\xi\dd\theta+\dd(i_\xi\theta)$ on the spacetime form indices,
\begin{align}
\delta J^\mu[\xi]
&=\delta\Pi^\mu{}_A\,X_\xi\phi^A+\Pi^\mu{}_A\,X_\xi\delta\phi^A-\xi^\mu\bigl(E_A\delta\phi^A+\partial_\nu\theta^\nu(\delta)\bigr) \nonumber\\
&=\omega^\mu(\delta,X_\xi)+\partial_\nu\bigl(\xi^\nu\Pi^\mu{}_A\delta\phi^A-\xi^\mu\Pi^\nu{}_A\delta\phi^A\bigr)-\xi^\mu E_A\delta\phi^A,
\end{align}
which is Eq.~\eqref{eq:master-identity}. The corner two-form $k^{\mu\nu}_\xi$ is one convenient representative, with equivalent choices differing by identically conserved superpotentials that move the boundary representative but not the identity. Integrating over $\Sigma$ on shell gives Eq.~\eqref{eq:integrated-identity}. The generator is Hamiltonian when the resulting corner one-form vanishes or is integrable on the chosen boundary phase space.

For reduced metric-sector STEGR, with $P^{\alpha\beta}=\eps_n n_\mu\bar\Pi^{\mu\alpha\beta}$ and the covariant split of $\nabla_{(\alpha}\xi_{\beta)}$ into descriptor derivatives and frame terms, the raw proper-measure scalars of Eq.~\eqref{eq:NC-split} are
\begin{align}
\mathcal C^{\rm STEGR}_{\perp,\rm raw}
&=P^{\alpha\beta}\bigl(\mathcal L_ng_{\alpha\beta}-Q_{\perp\alpha\beta}\bigr)-\bar\Pi^{i\alpha\beta}Q_{i\alpha\beta}+\Ham^{\rm bulk}_{\rm STEGR},
\\
\mathcal C^{\rm STEGR}_{i,\rm raw}
&=P^{\alpha\beta}\,\mathcal L_{e_i}g_{\alpha\beta},
\end{align}
with descriptor momenta
\begin{equation}
 \mathcal U^i_\perp=2P^{\alpha\beta}n_\alpha e^i{}_\beta,
 \qquad
 \mathcal U^i{}_j=2P^{\alpha\beta}e_{j\alpha}e^i{}_\beta,
\end{equation}
integrated by parts into the improved densities and the boundary term. The factor two follows from the symmetry of $P^{\alpha\beta}$ and the weighted symmetrization in $\mathcal L_\xi g_{\alpha\beta}$. The equivalence-restoring divergence of the action, $L^{\rm bdry}_{\rm STEGR}=-\partial_\alpha\Pi^\alpha$, changes the current by an exact form. With outward corner normal $s^\mu=s^ie_i{}^\mu$, its full corner is
\begin{equation}
B^{\rm STEGR,eq}_{\partial\Sigma}[N,N^i]
=\oint_{\partial\Sigma}dS\left[Ns_\mu\bar\Pi^\mu-(s_iN^i)P_{\rm tr}\right],
\qquad P_{\rm tr}:=\eps_n n_\mu\bar\Pi^\mu,
\end{equation}
a corner integral and not a piece of the bulk operator. The lapse-only expression follows when the shift preserves the corner, $s_iN^i=0$.

For field-independent $\xi$ and $\eta$, the active variations satisfy $[\delta_\xi,\delta_\eta]\phi^A=\delta_{[\xi,\eta]_{\rm Lie}}\phi^A$ up to internal gauge terms, giving the on-shell charge algebra with constraint and boundary remainders quoted in Sec.~\ref{ssec:algebra}. For fixed descriptors, varying $n^\mu$ under the first deformation, $\delta_\xi\eta^\mu=M\,\delta_\xi n^\mu+\delta_\xi(M^ie_i{}^\mu)$, the bracket that closes the variations is Eq.~\eqref{eq:BK}. In descriptors, this is the Dirac--Teitelboim algebra Eqs.~\eqref{eq:DT-lapse}--\eqref{eq:DT-shift}, equivalently the three elementary brackets
\begin{align}
[(0,\vec N),(0,\vec M)]=(0,\mathcal L_{\vec N}\vec M),\quad
[(0,\vec N),(M,0)]=(\mathcal L_{\vec N}M,0), \\
[(N,0),(M,0)]=\bigl(0,\,-\eps_n h^{ij}(N\partial_jM-M\partial_jN)\bigr).
\end{align}
The conditional Poisson closure in Eq.~\eqref{eq:conditional-closure} applies only after the corner flux is integrable. Nonintegrable flux instead requires the adjusted bracket of charge variations. In the compact reduced sector the boundary terms vanish and Eq.~\eqref{eq:reduced-closure} holds. No unreduced closure statement follows without adding the connection, multiplier, reducibility, and boundary sectors. Closure in nontrivial affine-holonomy sectors also requires flatness conditions on the holonomy lift.

%======================================================================
\section{Dirac analysis in the declared sector}
\label{app:dirac}

This appendix proves the classification in the sector used in the paper, following the standard Dirac procedure \cite{HenneauxTeitelboim1992}. Take a compact slice without boundary, subtract the STEGR equivalence boundary term, solve flatness and torsion locally in the trivial-holonomy sector, and impose the coincident gauge. Under these assumptions, the reduced bulk action has the ADM canonical form
\begin{align}
 S_{\mathrm{red}}={}&\int du\int_\Sigma d^3x\,
 \bigl(\pi^{ij}\dot\gamma_{ij}+\pi_N\dot N+\pi_i\dot N^i
 -N\mathcal H_\perp-N^i\mathcal H_i
 -\lambda_N\pi_N-\lambda^i\pi_i\bigr).
 \label{eq:app-dirac-action}
\end{align}
The canonical variables are $(N,N^i,\gamma_{ij};\pi_N,\pi_i,\pi^{ij})$, with fundamental brackets
\begin{align}
 \{\gamma_{ij}(x),\pi^{kl}(y)\}
 &=\delta^k_{(i}\delta^l_{j)}\delta_\Sigma(x,y), &
 \{N(x),\pi_N(y)\}&=\delta_\Sigma(x,y),\nonumber\\
 \{N^i(x),\pi_j(y)\}&=\delta^i_j\delta_\Sigma(x,y).
 \label{eq:app-dirac-brackets}
\end{align}
The absence of $\dot N$ and $\dot N^i$ from the original Lagrangian gives the four primary constraints $\pi_N\approx0$ and $\pi_i\approx0$. The total Hamiltonian and its smeared secondary constraints are
\begin{align}
 H_T&=H[N]+D[\vec N]+\int_\Sigma d^3x\,
 (\lambda_N\pi_N+\lambda^i\pi_i),\nonumber\\
 H[N]&=\int_\Sigma d^3x\,N\mathcal H_\perp,
\qquad  D[\vec N]=\int_\Sigma d^3x\,N^i\mathcal H_i.
\end{align}
Preservation of the primary constraints gives
\begin{equation}
 \dot\pi_N(x)=\{\pi_N(x),H_T\}=-\mathcal H_\perp(x)\approx0,
 \qquad
 \dot\pi_i(x)=\{\pi_i(x),H_T\}=-\mathcal H_i(x)\approx0.
 \label{eq:app-primary-stabilization}
\end{equation}
These are precisely the secondary constraints in Eq.~\eqref{eq:adm-constraints}.

On the compact slice, integration by parts gives
\begin{equation}
 D[\vec N]=\int_\Sigma d^3x\,\pi^{ij}\mathcal L_{\vec N}\gamma_{ij}.
\end{equation}
It follows directly from Eq.~\eqref{eq:app-dirac-brackets} that
\begin{equation}
 \{\gamma_{ij},D[\vec N]\}=\mathcal L_{\vec N}\gamma_{ij},
 \qquad
 \{\pi^{ij},D[\vec N]\}=\mathcal L_{\vec N}\pi^{ij},
 \label{eq:app-D-generator}
\end{equation}
where $\pi^{ij}$ is a tensor density of weight one. Thus $D[\vec N]$ generates spatial diffeomorphisms. This proves the first bracket below. The second follows because $\mathcal H_\perp$ is a scalar density of weight one.

It remains to calculate the bracket of two normal generators. The required functional derivatives are
\begin{align}
 \frac{\delta H[N]}{\delta\pi^{ij}}
 &=\frac{\kappa^2N}{\sqrt\gamma}
 \left(\pi_{ij}-\frac12\gamma_{ij}\pi\right),\nonumber\\
 \left.\frac{\delta H[N]}{\delta\gamma_{ij}}\right|_{D N}
 &=\frac{2\sqrt\gamma}{\kappa^2}
 \left(\gamma^{ij}D^2N-D^iD^jN\right).
 \label{eq:app-H-functionals}
\end{align}
The omitted terms contain no derivatives of the lapse. They are symmetric under $N\leftrightarrow M$ and cancel in the antisymmetrized bracket. Substitution of Eq.~\eqref{eq:app-H-functionals} gives
\begin{align}
 \{H[N],H[M]\}
 &=2\int_\Sigma d^3x\,\pi^{ij}
 \left(ND_iD_jM-MD_iD_jN\right)\nonumber\\
 &=\int_\Sigma d^3x\,(-2D_i\pi^{ij})
 \left(ND_jM-MD_jN\right)\nonumber\\
 &=D\!\left[\gamma^{ij}(N\partial_jM-M\partial_jN)\right].
 \label{eq:app-HH-calculation}
\end{align}
The second line uses one integration by parts and has no boundary contribution. Combining these results yields
\begin{align}
 \{D[\vec N],D[\vec M]\}&=D[[\vec N,\vec M]],\\
 \{D[\vec N],H[M]\}&=H[\mathcal L_{\vec N}M],\\
 \{H[N],H[M]\}&=D\!\left[h^{ij}(N\partial_jM-M\partial_jN)\right].
 \label{eq:dirac-reduced-algebra}
\end{align}
The secondary constraints do not depend on $N$, $N^i$, $\pi_N$, or $\pi_i$, so their brackets with the primary constraints vanish. Equation~\eqref{eq:dirac-reduced-algebra} also shows that preserving the secondary constraints produces only combinations of constraints and no tertiary constraints. Hence all eight constraints are first class. The phase space has dimension $20$ per point, so Eq.~\eqref{eq:dof} gives two local degrees of freedom. This completes the Dirac classification in the declared compact reduced sector.

For comparison, an unreduced action would keep independent $(g_{\mu\nu},\Gamma^\rho{}_{\mu\nu})$ and impose $T=R=0$ by multipliers. Its multiplier momenta, connection momenta, reducible flatness constraints, gauge generators, boundary pairs, and holonomy zero modes must all be stabilized before a rank can be quoted. Boundary conditions and holonomy strata can change the global constraint content, but no quantitative unreduced rank formula is proved here. All constraint counts in the paper refer to the compact reduced sector above.
%======================================================================
\end{appendices}

\bibliographystyle{unsrt}
\bibliography{sn-bibliography}

@article{BeltranJimenez:2019tjy,
  author  = {Beltr{\'a}n Jim{\'e}nez, Jose and Heisenberg, Lavinia and Koivisto, Tomi S.},
  title   = {The Geometrical Trinity of Gravity},
  journal = {Universe},
  volume  = {5},
  number  = {7},
  pages   = {173},
  year    = {2019},
  eprint  = {1903.06830},
  archivePrefix = {arXiv}
}

@article{Capozziello:2021pcg,
  author  = {Capozziello, Salvatore and Finch, Andrew and Said, Jackson Levi and Magro, Alessio},
  title   = {The 3+1 formalism in teleparallel and symmetric teleparallel gravity},
  journal = {Eur. Phys. J. C},
  volume  = {81},
  number  = {12},
  pages   = {1141},
  year    = {2021},
  eprint  = {2108.03075},
  archivePrefix = {arXiv}
}

@article{Sauer:2004hj,
  author  = {Sauer, Tilman},
  title   = {Field equations in teleparallel spacetime: Einstein's fernparallelismus approach towards unified field theory},
  journal = {Historia Math.},
  volume  = {33},
  pages   = {399--439},
  year    = {2006},
  eprint  = {physics/0405142},
  archivePrefix = {arXiv}
}

@article{1976Cho,
  author  = {Cho, Y. M.},
  title   = {Einstein Lagrangian as the translational Yang-Mills Lagrangian},
  journal = {Phys. Rev. D},
  volume  = {14},
  number  = {10},
  pages   = {2521--2525},
  year    = {1976}
}

@article{Hayashi1979,
  author  = {Hayashi, Kenji and Shirafuji, Takeshi},
  title   = {New general relativity},
  journal = {Phys. Rev. D},
  volume  = {19},
  pages   = {3524--3553},
  year    = {1979}
}

@article{Maluf:1994ji,
  author  = {Maluf, J. W.},
  title   = {Hamiltonian formulation of the teleparallel description of general relativity},
  journal = {J. Math. Phys.},
  volume  = {35},
  pages   = {335--343},
  year    = {1994}
}

@article{Maluf:2013gaa,
  author  = {Maluf, J. W.},
  title   = {The teleparallel equivalent of general relativity},
  journal = {Annalen Phys.},
  volume  = {525},
  pages   = {339--357},
  year    = {2013},
  eprint  = {1303.3897},
  archivePrefix = {arXiv}
}

@article{Nester:1998mp,
  author  = {Nester, James M. and Yo, Hwei-Jang},
  title   = {Symmetric teleparallel general relativity},
  journal = {Chin. J. Phys.},
  volume  = {37},
  pages   = {113},
  year    = {1999},
  eprint  = {gr-qc/9809049},
  archivePrefix = {arXiv}
}

@article{Blagojevic:2000pi,
  author  = {Blagojevic, M. and Vasilic, M.},
  title   = {Gauge symmetries of the teleparallel theory of gravity},
  journal = {Class. Quant. Grav.},
  volume  = {17},
  pages   = {3785--3798},
  year    = {2000},
  eprint  = {hep-th/0006080},
  archivePrefix = {arXiv}
}

@article{Blagojevic:2003cg,
  author  = {Blagojevic, M.},
  title   = {Three lectures on Poincare gauge theory},
  journal = {SFIN A},
  volume  = {1},
  pages   = {147--172},
  year    = {2003},
  eprint  = {gr-qc/0302040},
  archivePrefix = {arXiv}
}

@article{Barker2023,
  author  = {Barker, W. E. V.},
  title   = {Geometric multipliers and partial teleparallelism in Poincar\'e gauge theory},
  journal = {Phys. Rev. D},
  volume  = {108},
  pages   = {024053},
  year    = {2023},
  eprint  = {2205.13534},
  archivePrefix = {arXiv}
}

@article{HEHL19951,
  author  = {Hehl, Friedrich W. and McCrea, J. Dermott and Mielke, Eckehard W. and Ne'eman, Yuval},
  title   = {Metric-affine gauge theory of gravity: field equations, Noether identities, world spinors, and breaking of dilation invariance},
  journal = {Phys. Rep.},
  volume  = {258},
  number  = {1},
  pages   = {1--171},
  year    = {1995},
  eprint  = {gr-qc/9402012},
  archivePrefix = {arXiv}
}

@article{obukhov2003metric,
  author  = {Obukhov, Yuri N. and Pereira, Jos{\'e} G.},
  title   = {Metric-affine approach to teleparallel gravity},
  journal = {Phys. Rev. D},
  volume  = {67},
  number  = {4},
  pages   = {044016},
  year    = {2003},
  eprint  = {gr-qc/0212080},
  archivePrefix = {arXiv}
}

@article{BeltranJimenez:2017tkd,
  author  = {Beltr{\'a}n Jim{\'e}nez, Jose and Heisenberg, Lavinia and Koivisto, Tomi},
  title   = {Coincident General Relativity},
  journal = {Phys. Rev. D},
  volume  = {98},
  number  = {4},
  pages   = {044048},
  year    = {2018},
  eprint  = {1710.03116},
  archivePrefix = {arXiv}
}

@article{Capozziello:2022zzh,
  author  = {Capozziello, Salvatore and De Falco, Vittorio and Ferrara, Carmen},
  title   = {Comparing equivalent gravities: common features and differences},
  journal = {Eur. Phys. J. C},
  volume  = {82},
  number  = {10},
  pages   = {865},
  year    = {2022},
  eprint  = {2208.03011},
  archivePrefix = {arXiv}
}

@article{Hehl:1976my,
  author  = {Hehl, Friedrich W. and Kerlick, G. David and Von Der Heyde, Paul},
  title   = {On a New Metric Affine Theory of Gravitation},
  journal = {Phys. Lett. B},
  volume  = {63},
  pages   = {446--448},
  year    = {1976}
}

@book{Hehl2003,
  author    = {Hehl, Friedrich W. and Obukhov, Yuri N.},
  title     = {Foundations of Classical Electrodynamics: Charge, Flux, and Metric},
  series    = {Progress in Mathematical Physics},
  volume    = {33},
  publisher = {Birkh{\"a}user},
  address   = {Boston, MA},
  doi       = {10.1007/978-1-4612-0051-2},
  year      = {2003}
}

@article{Puetzfeld:2014qba,
  author  = {Puetzfeld, Dirk and Obukhov, Yuri N.},
  title   = {Equations of motion in metric-affine gravity: A covariant unified framework},
  journal = {Phys. Rev. D},
  volume  = {90},
  number  = {8},
  pages   = {084034},
  year    = {2014},
  eprint  = {1408.5669},
  archivePrefix = {arXiv}
}

@article{Obukhov:2018bmf,
  author  = {Obukhov, Yuri N.},
  title   = {Poincar{\'e} gauge gravity: An overview},
  journal = {Int. J. Geom. Meth. Mod. Phys.},
  volume  = {15},
  number  = {supp01},
  pages   = {1840005},
  year    = {2018},
  eprint  = {1805.07385},
  archivePrefix = {arXiv}
}

@article{Maluf:2000zt,
  author  = {Maluf, J. W. and da Rocha-Neto, J. F.},
  title   = {Hamiltonian formulation of general relativity in the teleparallel geometry},
  journal = {Phys. Rev. D},
  volume  = {64},
  pages   = {084014},
  year    = {2001},
  note    = {arXiv:gr-qc/0002059},
  eprint  = {gr-qc/0002059},
  archivePrefix = {arXiv}
}

@article{Guzman:2023,
  author  = {Guzm{\'a}n, Maria-Jose},
  title   = {The Hamiltonian constraint in the symmetric teleparallel equivalent of general relativity},
  journal = {Gen. Rel. Grav.},
  volume  = {58},
  number  = {3},
  pages   = {26},
  year    = {2026},
  note    = {arXiv:2311.01424},
  eprint  = {2311.01424},
  archivePrefix = {arXiv}
}

@misc{CapozzielloSauro,
  author = {Capozziello, Salvatore and Sauro, Dario},
  title  = {Extrinsic geometry and Hamiltonian analysis of symmetric teleparallel gravity},
  year   = {2026},
  note   = {arXiv:2604.19310},
  eprint  = {2604.19310},
  archivePrefix = {arXiv}
}

@article{DeWitt1967,
  author  = {DeWitt, Bryce S.},
  title   = {Quantum Theory of Gravity. I. The Canonical Theory},
  journal = {Phys. Rev.},
  volume  = {160},
  pages   = {1113--1148},
  year    = {1967}
}

@article{Isham:1992ms,
  author  = {Isham, C. J.},
  title   = {Canonical quantum gravity and the problem of time},
  journal = {NATO Sci. Ser. C},
  volume  = {409},
  pages   = {157--287},
  year    = {1993},
  eprint  = {gr-qc/9210011},
  archivePrefix = {arXiv}
}

@incollection{Kuchar:1991qf,
  author    = {Kucha{\v r}, Karel V.},
  title     = {Time and interpretations of quantum gravity},
  booktitle = {Proc. 4th Canadian Conf. on General Relativity and Relativistic Astrophysics},
  editor    = {Kunstatter, G. and Vincent, D. and Williams, J.},
  pages     = {211--314},
  publisher = {World Scientific},
  address   = {Singapore},
  year      = {1992}
}

@article{Teitelboim1973,
  author  = {Teitelboim, Claudio},
  title   = {How commutators of constraints reflect the space-time structure},
  journal = {Ann. Phys.},
  volume  = {79},
  pages   = {542--557},
  year    = {1973}
}

@article{Bergmann:1972ww,
  author  = {Bergmann, Peter G. and Komar, Arthur},
  title   = {The coordinate group symmetries of general relativity},
  journal = {Int. J. Theor. Phys.},
  volume  = {5},
  pages   = {15--28},
  year    = {1972}
}

@article{Hojman:1976vp,
  author  = {Hojman, Sergio A. and Kucha{\v r}, Karel and Teitelboim, Claudio},
  title   = {Geometrodynamics regained},
  journal = {Ann. Phys.},
  volume  = {96},
  number  = {1},
  pages   = {88--135},
  year    = {1976}
}

@article{Blohmann:2010jd,
  author  = {Blohmann, Christian and Fernandes, Marco Cezar Barbosa and Weinstein, Alan},
  title   = {Groupoid symmetry and constraints in general relativity},
  journal = {Commun. Contemp. Math.},
  volume  = {15},
  pages   = {1250061},
  year    = {2013},
  note    = {arXiv:1003.2857},
  eprint  = {1003.2857},
  archivePrefix = {arXiv}
}

@article{Bojowald:2016hgh,
  author  = {Bojowald, Martin and Brahma, Suddhasattwa and B{\"u}y{\"u}k{\c c}am, Umut and D'Ambrosio, Fabio},
  title   = {Hypersurface-deformation algebroids and effective space-time models},
  journal = {Phys. Rev. D},
  volume  = {94},
  number  = {10},
  pages   = {104032},
  year    = {2016},
  note    = {arXiv:1610.08355},
  eprint  = {1610.08355},
  archivePrefix = {arXiv}
}

@article{Gunther1987,
  author  = {G{\"u}nther, Christian},
  title   = {The polysymplectic Hamiltonian formalism in field theory and calculus of variations. I: the local case},
  journal = {J. Diff. Geom.},
  volume  = {25},
  pages   = {23--53},
  year    = {1987}
}

@article{Kanatchikov:1997wp,
  author  = {Kanatchikov, Igor V.},
  title   = {Canonical structure of classical field theory in the polymomentum phase space},
  journal = {Rept. Math. Phys.},
  volume  = {41},
  pages   = {49--90},
  year    = {1998},
  doi     = {10.1016/S0034-4877(98)80182-1},
  eprint  = {hep-th/9709229},
  archivePrefix = {arXiv}
}

@article{Kanatchikov:1997pj,
  author  = {Kanatchikov, Igor V.},
  title   = {On field theoretic generalizations of a Poisson algebra},
  journal = {Rept. Math. Phys.},
  volume  = {40},
  pages   = {225--234},
  year    = {1997},
  doi     = {10.1016/S0034-4877(97)85919-8},
  eprint  = {hep-th/9710069},
  archivePrefix = {arXiv}
}

@article{Kanatchikov:2000jz,
  author  = {Kanatchikov, Igor V.},
  title   = {Precanonical quantum gravity: Quantization without the space-time decomposition},
  journal = {Int. J. Theor. Phys.},
  volume  = {40},
  pages   = {1121--1149},
  year    = {2001},
  doi     = {10.1023/A:1017557603606},
  eprint  = {gr-qc/0012074},
  archivePrefix = {arXiv}
}

@misc{Kanatchikov:2023tegr,
  author = {Kanatchikov, Igor V.},
  title  = {The De Donder--Weyl Hamiltonian formulation of TEGR and its quantization},
  year   = {2023},
  note   = {arXiv:2308.10052},
  eprint  = {2308.10052},
  archivePrefix = {arXiv}
}

@article{Carinena:1991,
  author  = {Cari{\~n}ena, J. F. and Crampin, M. and Ibort, L. A.},
  title   = {On the multisymplectic formalism for first order field theories},
  journal = {Differ. Geom. Appl.},
  volume  = {1},
  number  = {4},
  pages   = {345--374},
  year    = {1991}
}

@misc{Gotay:1997,
  author = {Gotay, Mark J. and Isenberg, James and Marsden, Jerrold E. and Montgomery, Richard},
  title  = {Momentum Maps and Classical Relativistic Fields. Part I: Covariant Field Theory},
  year   = {2004},
  note   = {arXiv:physics/9801019},
  eprint  = {physics/9801019},
  archivePrefix = {arXiv}
}

@article{Forger:2004,
  author  = {Forger, Michael and Romero, Sandro Vieira},
  title   = {Covariant Poisson brackets in geometric field theory},
  journal = {Commun. Math. Phys.},
  volume  = {256},
  pages   = {375--410},
  year    = {2005},
  note    = {arXiv:math-ph/0408008},
  eprint  = {math-ph/0408008},
  archivePrefix = {arXiv}
}

@incollection{Wheeler1968,
  author    = {Wheeler, John A.},
  title     = {Superspace and the nature of quantum geometrodynamics},
  booktitle = {Battelle Rencontres},
  editor    = {DeWitt, C. and Wheeler, J. A.},
  pages     = {242--307},
  publisher = {W. A. Benjamin},
  address   = {New York},
  year      = {1968}
}

@incollection{Crnkovic1988CPS,
  author    = {Crnkovi{\'c}, {\v C}edomir and Witten, Edward},
  title     = {Covariant description of canonical formalism in geometrical theories},
  booktitle = {Three Hundred Years of Gravitation},
  editor    = {Hawking, S. W. and Israel, W.},
  pages     = {676--684},
  publisher = {Cambridge University Press},
  address   = {Cambridge},
  year      = {1987}
}

@article{Lee:1990nz,
  author  = {Lee, Joohan and Wald, Robert M.},
  title   = {Local symmetries and constraints},
  journal = {J. Math. Phys.},
  volume  = {31},
  pages   = {725--743},
  year    = {1990}
}

@article{Iyer:1994ys,
  author  = {Iyer, Vivek and Wald, Robert M.},
  title   = {Some properties of Noether charge and a proposal for dynamical black hole entropy},
  journal = {Phys. Rev. D},
  volume  = {50},
  pages   = {846--864},
  year    = {1994},
  eprint  = {gr-qc/9403028},
  archivePrefix = {arXiv}
}

@article{WaldZoupas:1999wa,
  author  = {Wald, Robert M. and Zoupas, Andreas},
  title   = {A General definition of ``conserved quantities'' in general relativity and other theories of gravity},
  journal = {Phys. Rev. D},
  volume  = {61},
  pages   = {084027},
  year    = {2000},
  eprint  = {gr-qc/9911095},
  archivePrefix = {arXiv}
}

@article{Regge:1974zd,
  author  = {Regge, Tullio and Teitelboim, Claudio},
  title   = {Role of surface integrals in the Hamiltonian formulation of general relativity},
  journal = {Ann. Phys.},
  volume  = {88},
  pages   = {286--318},
  year    = {1974}
}

@article{Brown:1992br,
  author  = {Brown, J. David and York, James W.},
  title   = {Quasilocal energy and conserved charges derived from the gravitational action},
  journal = {Phys. Rev. D},
  volume  = {47},
  pages   = {1407--1419},
  year    = {1993},
  note    = {arXiv:gr-qc/9209012},
  eprint  = {gr-qc/9209012},
  archivePrefix = {arXiv}
}

@book{HenneauxTeitelboim1992,
  author    = {Henneaux, Marc and Teitelboim, Claudio},
  title     = {Quantization of Gauge Systems},
  publisher = {Princeton University Press},
  address   = {Princeton, NJ},
  isbn      = {978-0-691-03769-1},
  year      = {1992}
}

@article{Utiyama,
  author={Utiyama, Ryoyu},
  title={Invariant theoretical interpretation of interaction},
  journal={Phys. Rev.}, volume={101}, pages={1597--1607}, year={1956}
}

@article{Sardanashvily:1994fg,
  author={Sardanashvily, Gennadi},
  title={Hamiltonian field systems on composite manifolds},
  year={1994}, eprint={hep-th/9409159}, archivePrefix={arXiv}
}

@article{Obukhov:2023yti,
  author={Obukhov, Yuri N. and Hehl, Friedrich W.},
  title={Hyperfluid model revisited},
  journal={Phys. Rev. D}, volume={108}, number={10}, pages={104044}, year={2023},
  eprint  = {2308.06598},
  archivePrefix = {arXiv}
}

@article{Isham:1984sb,
  author={Isham, Christopher J. and Kuchar, Karel V.},
  title={Representations of space-time diffeomorphisms. 1. Canonical parametrized field theories},
  journal={Annals Phys.}, volume={164}, pages={288--315}, year={1985}
}

@incollection{ArnowittDeserMisner1962,
  author    = {Arnowitt, Richard and Deser, Stanley and Misner, Charles W.},
  title     = {The Dynamics of General Relativity},
  booktitle = {Gravitation: An Introduction to Current Research},
  editor    = {Witten, Louis},
  publisher = {Wiley},
  address   = {New York},
  pages     = {227--265},
  year      = {1962},
  eprint    = {gr-qc/0405109},
  archivePrefix = {arXiv}
}

\end{document}